\documentclass[a4paper,11pt]{article}

\usepackage[T1]{fontenc}

\usepackage{ae,aecompl}
\usepackage{amsfonts}
\usepackage{amssymb}
\usepackage{amsmath}
\usepackage{amsthm}
\usepackage{verbatim}
\usepackage{tikz}
\usepackage{graphicx}
\usepackage{hyperref}
\usepackage{xcolor}
\usepackage{xspace}
\usepackage{enumitem}
\usepackage{authblk}
\usepackage[font=scriptsize]{caption}
\usepackage{vwcol}
\usepackage{arydshln}
\theoremstyle{plain}
\theoremstyle{plain}
\newtheorem{theorem}{Theorem}[section]
\newtheorem{lemma}[theorem]{Lemma}
\newtheorem{corollary}[theorem]{Corollary}
\newtheorem{proposition}[theorem]{Proposition}
\theoremstyle{definition}
\newtheorem{definition}[theorem]{Definition}
\newtheorem{remark}[theorem]{Remark}
\newtheorem{example}[theorem]{Example}

 \usepackage{vwcol}  

\newcommand{\aamasp}[1]{{\color{magenta} #1}}
\newcommand{\techrep}[1]{{\color{blue} #1}}
\newcommand{\vcut}[1]{}
\newcommand{\anote}[1]{{\footnotesize \color{blue}\bf{}{AK: #1}}}

\newcommand{\rnote}[1]{{\footnotesize \color{green}\bf{}{RR: #1}}}

\newcommand{\vnote}[1]{}

\newcommand{\Logicname}[1]{\ensuremath{\mathsf{#1}}}

\newcommand{\WLC}{\Logicname{WLC}\xspace}
\newcommand{\CM}{\Logicname{CM}\xspace}
\newcommand{\WM}{\Logicname{WM}\xspace}
\newcommand{\LA}{\Logicname{LA}\xspace}

\newcommand{\ECT}{\Logicname{ECT}\xspace}
\newcommand{\GCT}{\Logicname{GCT}\xspace}
\newcommand{\OSCP}{\Logicname{OSCP}\xspace}

\newcommand{\defstyle}{\textbf}

\usepackage[margin=1in]{geometry}

\author[1]{Antti Kuusisto}
\author[2]{Raine R\"{o}nnholm}
\affil[1]{University of Helsinki and Tampere University, Finland}

\affil[2]{Universit\'{e} Paris-Saclay, France}

\date{}

\begin{document}

\title{Optimal protocols for the most difficult \\ repeated coordination games}

\maketitle

\vspace{2.5cm}

\begin{center}
    \textbf{Abstract}
\end{center}

\noindent
This paper investigates repeated win-lose
coordination games (\WLC-games). We analyse which
protocols are optimal for these games
covering both the worst case and
average case scenarios, i,e., optimizing the
guaranteed and expected coordination times. We begin by
analysing Choice Matching Games (\CM-games) which are a
simple yet fundamental type of \WLC-games, where the
goal of the players is to pick the same choice from a
finite set of initially indistinguishable choices. We
give a complete classification of optimal expected and
guaranteed coordination times in two-player \CM-games
and show that the corresponding optimal protocols are
unique in every case---except in the \CM-game
with four choices, which we analyse separately. 

Our results on \CM-games are also essential for
proving a more general result on the
difficulty of all \WLC-games: we provide a
complete analysis of least upper bounds for
optimal expected coordination times in all
two-player \WLC-games as a function of 
game size. We also show that \CM-games can be
seen as the most difficult games
among all two-player \WLC-games, as they turn out to 
have the greatest optimal expected coordination times.

\medskip

\medskip


\newpage


\section{Introduction}
\label{sect:intro}
Pure win-lose coordination games (\WLC-games) are 
simple yet fundamental games  
where all players receive the same payoffs: 1 (win) or 0 (lose).
This paper studies \emph{repeated} \WLC-games, where the players make
simultaneous choices in discrete rounds
until (if ever) succeeding to coordinate on a winning profile.
\emph{Choice matching games} (\CM-games) are the simplest 
class of such games. The choice matching game $\CM_m^n$ has
$n$ players with the goal to choose the same choice 
among $m$ different indistinguishable choices, with no
communication during play. The players 
can use the history of the game (i.e.,
the players' choices in different rounds) for their benefit as the game proceeds.
For simplicity, we denote the two-player game $\CM_m^2$ by $\CM_m$.
%
%
%
%
%

A paradigmatic real-life scenario with a choice matching game
relates to a phenomenon that has humorously been called
``pavement tango'' or ``droitwich'' in \cite{droit}. Here two people try to 
pass each other but may end up 
blocking each other by repeatedly 
moving sideways into the same direction.
For another example of a choice matching game, consider $\CM_3$, 
the coordination-based 
variant of the rock-paper-scissors game, pictured on the right.

\begin{vwcol}[widths={0.85,0.15}, sep=2mm, justify=flush, rule=0pt, lines=6] 
\indent
Here the
two players (i.e., columns) coordinate if they succeed 
choosing an edge from one of the three rows.
The players first choose randomly; suppose they 
select the nodes in dotted circles. Simply based on symmetries, it then
makes sense for both players to choose from the last row (solid circles), as
each of the two other 
choices in each column have a symmetric, non-coordinating choice in the other column.
This leads to coordination in the second round.

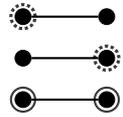
\begin{tikzpicture}[scale=0.55,choice/.style={draw, circle, fill=black!100, inner sep=2.1pt},
	location1/.style={draw, very thick, densely dotted, circle, color=black!77, inner sep=3.2pt},
	location2/.style={draw, thick, circle, color=black!77, inner sep=3.2pt}]
	\node at (0,0) [choice] (00) {};
	\node at (2,0) [choice] (20) {};
	\node at (0,1) [choice] (01) {};	
	\node at (2,1) [choice] (21) {};
	\node at (0,2) [choice] (02) {};	
	\node at (2,2) [choice] (22) {};
	\node at (02) [location1] {};
	\node at (21) [location1] {};
	\node at (00) [location2] {};
	\node at (20) [location2] {};
	\draw[thick] (00) to (20);
	\draw[thick] (01) to (21);
	\draw[thick] (02) to (22);	
%
%
%
	\node at (0,3) {};	
\end{tikzpicture}
%
\end{vwcol}

A general $n$-player \WLC-game is a generalization of $\CM_m^n$ where the
players do not necessarily have to choose from the same row to coordinate,
and it may not even suffice to choose from the same row. In classical 
matrix form representation, two-player choice matching games have ones on
the diagonal and zeroes elsewhere, while general two-player \WLC-games have general
distributions of ones and zeroes;
see Definition \ref{definition: games} for the full formal details.

In repeated \WLC-games, it is natural to try to coordinate as quickly as possible. 
There are two main scenarios to be investigated: \emph{guaranteeing
coordination} (with certainty) in as few rounds as possible and  
minimizing the \emph{expected number of rounds} for coordination.
The former concerns the number of
rounds it takes to coordinate in the \emph{worst case} and is measured in 
terms of \emph{guaranteed coordination times} (GCTs). The latter
relates to the \emph{average case} analysis measured in
terms of \emph{expected coordination times} (ECTs).

\paragraph{Our contributions.}

We provide a comprehensive study of upper
bounds for coordination in \emph{all} two-player
repeated \WLC-games, including a classification of related
optimal strategies (called \emph{protocols} in this work).
\CM-games are central to our work,
being a fundamental class of games and also the 
most difficult games for coordination---in a sense made precise below.

Two protocols play a central role in our study.
We introduce the so-called \emph{loop avoidance} 
protocol \LA (cf. Definition \ref{def: RES}) that essentially tells players to
play so that the generated history of choices 
always reduces the symmetries (e.g., automorphisms) of the game structure. 
We also use the so-called \emph{wait-or-move} (\WM)
protocol (cf. Definition \ref{definition: waitormove}),
essentially telling players to randomly
alternate between two choices that both 
coordinate with at least one of the opponent's two choices. We show
that \WM leads to coordination in
\emph{all} \WLC-games \emph{very fast},
the \ECT being $3-2p$,
where $p$ is the probability of
coordinating in the first round with random choices.

We then provide a complete analysis of the optimal {\ECT}s and {\GCT}s in all
choice matching games $\CM_m$. We also identify the 
protocols giving the
optimal {\ECT}s and {\GCT}s and show their
\emph{uniqueness}, where possible.
The table in Figure~\ref{fig: summary table} summarizes these results.
%
%
%
%
%
%
%
%
%
%
\setlength{\belowcaptionskip}{-20pt}
\begin{figure}[htp]
\begin{center}
\scalebox{0.8}{
{\footnotesize
\begin{tabular}{|c||c|c||c|c|}
 \hline
& Optimal \textbf{expected} & Unique optimal & Optimal \textbf{guaranteed} & Unique optimal \\
$m$ & coordination & protocol for & coordination & protocol for \\
& time in $\CM_m$ & expected time & time in $\CM_m$ & guaranteed time \\
\hline \rule{0pt}{1\normalbaselineskip} 
$1$\; & 1 & (any) & 1 & (any) \\[1mm]
$2$ & $2$ & WM & $\infty$ & --- \\[1mm]
$3$ & $1+\frac{2}{3}$ & \LA & $2$ & \LA \\[1mm]
$4$ & $2+\frac{1}{2}$ & --- & $\infty$ & --- \\[1mm]
$5$ & $2+\frac{1}{3}$ & \LA & $3$ & \LA \\[1mm]
\hdashline \rule{0pt}{1\normalbaselineskip} 
$6$\; & $2+\frac{2}{3}$ & WM & $\infty$ & --- \\[1mm]
$7$ & $2+\frac{5}{7}$ & WM & $4$ & \LA \\
\vdots & \vdots & \vdots & \vdots & \vdots \\
$2k$ & $3 - \frac{1}{k}$ & WM & $\infty$ & --- \\[1mm]
$2k+1$ & $3 - \frac{2}{2k+1}$ & WM & $k$ & \LA \\[1mm]
\hline
\end{tabular}
}
}
\caption{A complete analysis of two-player choice matching games.}
\label{fig: summary table}
\end{center}
\end{figure}
%
%
%
%
%
%
%
%
%
%
%
%
%
%
%
This analysis is complete, as we prove that there 
exists a continuum of optimal protocols for $\CM_4$ and establish that
for all even $m$, no protocol \emph{guarantees} a win in $\CM_m$.

Concerning the more general class of all \WLC-games, we
provide the following complete characterization of
upper bounds for the optimal {\ECT}s in
all two-player $\WLC$-games 
\emph{as a function of game size} 
(a game in a classical matrix form is of size $m$ when the maximum of the
number of rows and columns is $m$):


\medskip

\noindent
\textbf{Theorem.}
\emph{For any $m$, the greatest optimal \ECT among 
two-player $\WLC$-games of size $m$ is as follows:}
\vspace{-5mm}
\begin{center}
\scalebox{0.8}{
\begin{tabular}{|c|c|c|c|}
\hline Game size
& $m\in\mathbb{Z}_+\setminus\{3,5\}$ & \quad$m=5$\quad\text{} & $m=3$ \\
\hline \rule{0pt}{1.2\normalbaselineskip} 
Greatest optimal\; \ECT & $3-\frac{2}{m}$ & $2+\frac{1}{3}$ & $\frac{1+\sqrt{4+\sqrt{17}}}{2}$ $(\approx 1,925)$ \\[1mm]
\hline
\end{tabular}
}
\end{center}

\smallskip

\noindent
Also, concerning two-player choice matching games,
we establish that $\CM_m$ has the
strictly greatest optimal \ECT out of all two-player $\WLC$-games of 
size $m\not = 3$, making \CM-games the most difficult \WLC-games to 
coordinating in. We give a separate full analysis of the case $m=3$.

\paragraph{Related work.}
Coordination games (see, e.g., \cite{russellcooper},
\cite{Biglaiser1994}) 
are a key topic in game theory, 
with the early foundations laid, inter alia, in the works of
Schelling \cite{schelling} and Lewis \cite{Lewis69}.
Repeated games are---likewise---a key
topic, see for example \cite{mertens}, \cite{AumannMaschler95}, \cite{MailathSamuelson06}.
For seminal work on \emph{repeated coordination games},
see for example the articles \cite{CrawfordHaller95}, \cite{crawfordadaptive},
\cite{lagunoffmatsui}.

%
%

However, \WLC-games are a simple class of games that have not
been extensively studied in the literature. In particular,
choice matching games clearly
constitute a \emph{fundamental class of games},
and it is thus surprising
that the analysis of the current paper has not
been previously carried out.
Thus the related analysis is well
justified; it \emph{closes an obvious gap} in the literature.

In general, our study differs from the classical game-theoretic 
study of repeated games where
the focus is on
accumulated payoffs. Indeed, our repeated \WLC-games are based on 
\emph{reachability objectives}.
Especially our
worst case analysis (but also the average case study)
has only superficial overlap with
most work on repeated games.

However, similar work exists, the most notable example being
the seminal article \cite{CrawfordHaller95} that studies a
generalization of \WLC-games in a framework that has some similarities
with our setting. They introduce (what is 
equivalent to) the two-player \CM-games in their
final section on general examples.
They also essentially identify the optimal ways of playing $\CM_2$ and $\CM_3$, 
discussed also in this article, although in a technically
somewhat different setting of accumulated payoffs. Furthermore, 
they observe that a protocol essentially equivalent to $\WM$ is the best
way to play $\CM_6$, an observation we also make in our
setting. However, optimality of $\WM$ in $\CM_6$ is
not proved in \cite{CrawfordHaller95}. This would require an 
extensive analysis proving that the players cannot make beneficial 
use of asymmetric histories created by non-coordinating choices.
Indeed, the main technical difficulty in our corresponding
setting is to show \emph{uniqueness} of the optimal protocol.

Nonetheless, despite the differences, the framework of \cite{CrawfordHaller95}
bears some conceptual similarities to ours, e.g., the authors also 
identify \emph{structural protocols} (cf. Definition 
\ref{structuralprotocoldfn} below) as
the natural notion of strategy for studying their framework.
Furthermore, they make extensive use of focal points \cite{schelling} in 
analysing how asymmetric histories can potentially be used for coordination.
%

Relating to uniqueness of protocols, 
\cite{Goyal96} argues that individual rationality considerations not are not sufficient for players to ``learn how to coordinate'' in the setting of \cite{CrawfordHaller95}. We agree with \cite{Goyal96} that some
conventions are needed if several protocols lead to the optimal result.
However---in our framework---since we can prove
\emph{uniqueness} of the optimal protocols for $\CM_m$ (when $m\neq 4$), then
arguably rational players should adopt precisely these
protocols in \CM-games.

\paragraph{Techniques used.}
Some of our results are of course based on 
massaging techniques from 
game theory and mathematical analysis to suit our purposes. This
involves the standard things: infinite series, analysis of extrema, et cetera.
However, the core of our work relies on an original approach to
games based on relational structures, as
opposed to using the traditional matrix form representation.
This approach enables us to use graph theoretic ideas in our arguments.

Both in the worst-case and in average-case analysis, the main 
technical work relies heavily on analysis of symmetries---especially
the way the groups of
automorphisms of games evolve when playing coordination games.
The most involved result of the worst-case analysis,
Theorem \ref{the: odd m}, is 
proved by reducing the cardinality of the
automorphism group of the \WLC-game studied in a maximally fast fashion.
In the average-case analysis, Theorems \ref{the: 6-ladders},
\ref{the: 5-ladders}, \ref{the: 4-ladders} are proved via a 
combination of analysis of extrema; keeping track of 
groups of automorphisms; graph theoretic methods; and
focal points \cite{schelling} for breaking symmetry.
The most demanding part here is to show
\emph{uniqueness} of the protocols involved.
Also in the average-case analysis,
Theorem \ref{the: upper bounds for ECTs} relies on earlier theorems and an
extensive and exhaustive analysis of certain bipartite graphs.

We first used our approach to games via
relational structures in \cite{lori2017}, \cite{eumas}. It
has been applied also in the repeated setting in \cite{ecai2020} and 
considered in a more general
setting in \cite{comp19}, \cite{double}.

\section{Preliminaries}
\label{sec:prelim}

We define win-lose coordination games as \emph{relational structures}
the same way as in \cite{lori2017}, \cite{eumas}, \cite{ecai2020}: 
%

\begin{definition}\label{definition: games}
An $n$-player \defstyle{win-lose coordination game (\WLC-game)}
is a relational structure $G=(A,C_1,\dots,C_n,W_G)$ where $A$ is a finite
domain of \defstyle{choices}, each $C_i$ is a non-empty unary relation (representing the choices of player $i$) such that $C_1\cup\cdots\cup C_n=A$,  
and $W_G\subseteq C_1\times\cdots\times C_n$ is an $n$-ary \defstyle{winning relation}.  
For technical convenience, we assume that the players have pairwise disjoint choice sets, i.e., $C_i\cap C_j=\emptyset \text{ for every } i,j\leq n \text{ such that } i\neq j$.
A tuple $\sigma\in C_1\times\cdots\times C_n$ is called a \defstyle{choice profile} for $G$
and the choice profiles in $W_G$ are called \defstyle{winning choice profiles}.
We assume that there are no surely losing choices, i.e., choices $c\in A$ that do not belong to any winning choice profile, as rational players would never select such choices. 
The \defstyle{complement $\overline G$} of $G$ is
defined as $\overline G :=(A,C_1,\dots,C_n,\ C_1
\times \cdots\times C_n\, \setminus W_G)$.
\end{definition}

We will use the visual representation of \WLC
games as hypergraphs from \cite{lori2017}; two-player games 
become just bipartite graphs under this scheme.
The choices of each player are displayed as
columns of nodes, starting from the choices of
player 1 on the left and ending with the column of choices of player $n$. The
winning relation consists of lines that represent the winning choice profiles. Thus winning choice profiles are also
called \textbf{edges}.
See Example \ref{drawingexamples} in Appendix~\ref{appendix: Examples}
for an illustration of
the drawing scheme.



Consider a \WLC-game $G=(A,C_1,\dots,C_n,W_G)$
with $n$ players and $m$ winning 
choice profiles that do not intersect, i.e., none of the $m$ winning choice
profiles share a choice $c\in A$.
Such games form a simple yet 
fundamental and natural class of games, where the goal of
the players is simply to pick the same ``choice'', i.e., to 
simultaneously pick one of the $m$ winning profiles.
These games are called \defstyle{choice matching games}. We let $\CM^n_m$
denote the choice matching game with $n$ players and $m$
choices for each player.
%
%
In this article, we extensively make use of the two-player
choice matching games, $\CM^2_m$.
For these games, we will omit the
superscript ``$2$'' and simply denote them by $\CM_m$.
(Recall here the example $\CM_3$
pictured in the introduction.)

Interestingly, out of
all $n$-player \WLC-games where
each of the $n$ players has $m$ choices, the game $\CM^n_m$ has the
least probability of coordination when each player plays randomly. 
In this sense these games can be seen the most difficult for coordination. 
A fully compelling reason for the maximal difficulty
of choice matching games is given later on by Corollary~\ref{hardest games}.

\section{Repeated \WLC-games}
\label{sect:RWLC}

A \defstyle{repeated play of a \WLC-game $G$} 
consists of consecutive (one-step) plays of $G$.
The repeated play is continued until the
players successfully coordinate, i.e.,
select their choices from a winning choice profile. This
may lead to infinite plays.
We assume that each player can
remember the full history of the
repeated play and use this information when planning
the next choice. The history of the play after $k$ rounds is
encoded in a sequence $\mathcal{H}_k$
defined as follows.

\begin{definition}\label{definition: stages}
Let $G$ be an $n$-player \WLC-game.
A pair $(G,\mathcal{H}_k)$ is
called a \defstyle{stage $k$ (or $k$th stage)
in a repeated play of $G$},
where the \defstyle{history} $\mathcal{H}_k$ is a $k$-sequence of
choice profiles in~$G$.
More precisely, $\mathcal{H}_k=\big(H_i\big)_{i\in\{1,\dots,k\}}$
where each $H_i$ is an $n$-ary
relation $H_i=\{(c_1,\dots,c_n)\}$ with a 
single tuple $(c_1,\dots,c_n)\in C_1\times\cdots\times C_n$.
In the case $k=0$, we define $\mathcal{H}_0=\emptyset$. The
stage $(G,\mathcal{H}_0)$ is
the \defstyle{initial stage} (or the $0$th stage).
Like $G$, also $(G,\mathcal{H}_k)$ is a relational structure.
\end{definition}


A stage $k$ contains a history specifying precisely $k$ choice profiles chosen in a repeated play. 
%
%
A winning profile of $(G,\mathcal{H}_k)$ is
called a \defstyle{touched edge} if it
contains some choice $c$ picked in some
round $1,\dots , k$ leading to $(G,\mathcal{H}_k)$.
%
%
%
As we assume that the players only need to coordinate once, we consider repeated plays only up to the first stage where some winning choice profile is selected. If
coordination occurs in the $k$th round, then the $k$th stage is called the \defstyle{final stage} of the repeated play.
But a play can indeed possibly take infinitely long without coordination.

\begin{vwcol}[widths={0.8,0.2}, sep=5mm, justify=flush, rule=0pt,lines=5] 
\indent
On the right is a drawing of the stage $2$ in a repeated play of $\CM_2$,
the ``coordination game variant'' of the \emph{matching pennies game} 
(or the ``pavement tango'' from the introduction). 
Here the players have failed to coordinate in round 1 (having picked the choices
with dotted circles) and then failed again by both swapping their choices in
round 2 (solid circles).

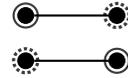
\begin{tikzpicture}[scale=0.6,choice/.style={draw, circle, fill=black!100, inner sep=2.1pt},
	location1/.style={draw, very thick, densely dotted, circle, color=black!77, inner sep=3.2pt},
	location2/.style={draw, thick, circle, color=black!77, inner sep=3.2pt}]
	\node at (0,0) [choice] (00) {};
	\node at (2,0) [choice] (20) {};
	\node at (0,1) [choice] (01) {};	
	\node at (2,1) [choice] (21) {};
	\node at (00) [location1] {};
	\node at (21) [location1] {};
	\node at (01) [location2] {};
	\node at (20) [location2] {};
	\draw[thick] (01) to (21);
	\draw[thick] (00) to (20);	
%
%
    \node at (0,2.3) {};
\end{tikzpicture}
\end{vwcol}

%

We next generalize the definition of \emph{protocols} from \cite{lori2017}.
In the current paper, a protocol
describes a \emph{mixed strategy for all stages in all \WLC-games and
for all player roles $i$}: 

\begin{definition}\label{def: Protocols}
A \defstyle{protocol} $\pi$ is a function outputting a probability distribution $f: C_i\rightarrow[0,1]$ (where $\sum_{c\in C_i}f(c)=1$) with the input of a player $i$ and a stage $(G,\mathcal{H}_k)$ of a repeated \WLC-game.
\end{definition}
Since a protocol can depend on the full history of the current stage, it gives a mixed, memory-based strategy for any repeated \WLC-game. Thus protocols can informally be regarded as global ``behaviour styles" of agents
over the class of all repeated \WLC-games.
It is important note that all players can see (and remember) the previous choices selected by all the other players---and also the order in which the 
choices have been made.

In the scenario that we study, it is obvious to require that the protocols should act \emph{independently of the names of choices and the names (or ordering) of player roles $i$}.\footnote{Note that if this assumption is not made, then coordination can trivially be guaranteed in a single round in any \WLC game by using a protocol which chooses some winning choice profile with probability $1$.}
In \cite{CrawfordHaller95}, this requirement follows from the ``assumption of no common language'' (for describing the game),
and in \cite{lori2017}, we say that such protocols are \emph{structural}.
To extend this concept for repeated games, we
first need to define the notion of a \emph{renaming}. 
The intuitive idea of 
renamings is to extend \emph{isomorphisms} between game 
graphs---including the history---to additionally enable
\emph{permuting the players} $1,\dots , n$ (see Example~\ref{ex: Renaming} in
Appendix~\ref{appendix: Examples} for an illustration of the definition).

\begin{definition}[Cf. \cite{lori2017}]\label{def: renamings}
%
A \defstyle{renaming} between
stages $(G,\mathcal{H}_k)$ and $(G',\mathcal{H}_k')$ of $n$-player \WLC-games $G$ and $G'$ is a
pair $(\beta,\pi)$ where $\beta$ is a permutation of $\{1,\dots , n\}$ and $\pi$ a
bijection from the domain of $G$ to that of $G'$ such that

\smallskip

$\bullet$ $c\in C_{\beta(i)}\Leftrightarrow \pi(c) \in C_i'$ for all $i\leq n$ and $c$ in the domain of $G$,

$\bullet$ $(c_1,\dots , c_n) \in W_G
\Leftrightarrow (\pi(c_{\beta(1)}),\dots , \pi(c_{\beta(n)}))\in W_G'$,

$\bullet$ $(c_1,\dots , c_n) \in H_i
\Leftrightarrow (\pi(c_{\beta(1)}),\dots , \pi(c_{\beta(n)}))\in H_i'$ for all $i\leq k$.

\smallskip

If $(G,\mathcal{H}_k)$ and $(G',\mathcal{H}_k')$ have the
same domain $A$, we say that $(\beta,\pi)$ is a \defstyle{renaming of $(G,\mathcal{H}_k)$}.
Choices $c\in C_i$ and $d\in C_j$ are \defstyle{structurally equivalent}, denoted by
$c\sim d$, if there is a renaming $(\beta,\pi)$ of $(G,\mathcal{H}_k)$ such that $\beta(i)=j$ and $\pi(c)=d$.
It is easy to see that $\sim$ is an equivalence relation on $A$. We denote the equivalence class of a choice $c$ by $[c]$.
\end{definition}

\begin{definition}\label{structuralprotocoldfn}
A protocol $\pi$ is \defstyle{structural} if it is indifferent with
respect to renamings, meaning that if $(G,\mathcal{H}_k)$ and $(G',\mathcal{H}_k')$ 
are stages with a renaming $(\beta,\pi)$ between them,
then for any $i$ and any $c\in C_i$, we have 
$
	f(c)=f'(\pi(c)),
$
where $f= \pi((G,\mathcal{H}_k),i)$
and $f' = \pi((G',\mathcal{H}_k'),\beta(i))$.

\end{definition}

%
Note that a structural protocol may depend on the full
history, which records even the order in which the choices have been played.
%
%
%
%
Hereafter we assume all protocols to be structural.

\begin{definition}\label{def: similarity}
Let $G$ be a \WLC-game and let $S$ and $S'$ be stages of $G$. Let $\sim$
(respectively $\sim'$) be the structural equivalence relation over $S$ (respectively, $S'$).
We say that $S$ and $S'$ are \defstyle{automorphism-equivalent} if $\sim \,=\, \sim'$.
The stages $S$ and $S'$ are \defstyle{structurally similar} if one can be obtained
from the other by a chain of renamings and automorphism-equvalences.
\end{definition}


%
A choice $c$ in a
stage $S$ is a \defstyle{focal point} if it is not
structurally equivalent to any other choice in
that same stage $S$,
with the possible exception of 
choices $c'$ belonging to a same edge as $c$. 
See Example~\ref{ex: focal points} for an illustration of focal points.
%
A focal point breaks symmetry and can be used for winning a 
repeated coordination game. This requires that the players have some
(possibly prenegotiated) way to choose some
edge $(u,v)$ such that $u$ or $v$ is a focal point.
%

In repeated coordination games, it is natural to try to
\emph{coordinate as quickly as possible}.
There are two principal scenarios related to optimizing coordination times:
the \emph{average case} and the \emph{worst case}. The
former concerns the expected
number rounds for coordination and the latter the maximum number in
which coordination can be guaranteed with certainty.

\begin{definition}
Let $(G,\mathcal{H}_k)$ a stage and let $\pi$ be a protocol.
%
The \defstyle{one-shot coordination probability (\OSCP) from $(G,\mathcal{H}_k)$ with $\pi$} is the probability of
coordinating in a single round from $(G,\mathcal{H}_k)$ when each player follows $\pi$.
%
The \defstyle{expected coordination time (\ECT) from $(G,\mathcal{H}_k)$ with $\pi$} is the
expected value for the number of rounds until
coordination from $(G,\mathcal{H}_k)$ when all players follow $\pi$.
The \defstyle{guaranteed coordination time (\GCT) from $(G,\mathcal{H}_k)$ with $\pi$} is the number $n$ such
that the players are \emph{guaranteed} to coordinate from $(G,\mathcal{H}_k)$ in $n$ rounds, but not in $n-1$ rounds,
when all players follow $\pi$, if such a
number exists. Else this value is $\infty$.
%
%
%

The \OSCP, \ECT and \GCT
from the initial stage $(G,\emptyset)$ with $\pi$ are
referred to as the \OSCP, \ECT and \GCT in $G$ with $\pi$.
We say that $\pi$ is \defstyle{\ECT-optimal} for $G$ if $\pi$
gives the minimum \ECT in $G$, i.e.,  the \ECT given by any protocol $\pi'$ is at
least as large as the one given by $\pi$.
\defstyle{\GCT-optimality} of $\pi$ for $G$ is defined analogously.
\end{definition}



It is possible that there are several different
protocols giving the optimal \ECT (or \GCT)
for a given \WLC-game.
If two protocols $\pi_1$ and $\pi_2$ are both optimal, it may
be that the optimal value is nevertheless \emph{not} obtained when
some of the players follow $\pi_1$ and the
others $\pi_2$.
This leads to a meta-coordination problem about choosing the same
optimal protocol to follow. 
However, such a problem will be avoided if there exists a unique optimal protocol.

\begin{definition}
Let $\pi$ be a protocol and $G$ a \WLC-game.
We say that $\pi$ is \defstyle{uniquely \ECT-optimal} for $G$ if $\pi$ is \ECT-optimal for $G$ and the following holds for all
other protocols $\pi'$ that are \ECT-optimal for $G$: 
for any stage $S$ in $G$ that is reachable with $\pi$, we have $\pi'(S)=\pi(S)$.  
\defstyle{Unique \GCT-optimality} of $\pi$ for $G$ is defined analogously.\footnote{Note that if two different
protocols are uniquely \ECT-optimal for G (and
similarly for unique \GCT-optimality), then their behaviour on $G$
can differ only on stages that are not reachable in the first place by the protcols. Also, their behaviour can of
course differ on games other than $G$.}
\end{definition}

The next lemma states that two structurally similar
stages are essentially the same
stage with respect to different {\ECT}s and {\GCT}s. 
The proof is straightforward.

\begin{lemma}\label{similarity lemma}
Assume stages $S$ and $S'$ of $G$ are structurally similar.
Now, for any protocol $\pi$, 
there exists a protocol $\pi'$ which
gives the same \ECT and \GCT from $S'$ as $\pi$ gives from $S$.
\end{lemma}

\section{Protocols for repeated \WLC-games}
\label{sect:solvability}

In this section we introduce two special protocols,
the \emph{loop avoidance protocol} \LA and the
\emph{wait-or-move protocol} \WM. Informally, \LA asserts
that in every round, every player $i$
should avoid---if possible---all choices $c$ that could possibly make
the resulting stage
automorphism-equivalent 
(cf. Def.~\ref{def: similarity}) 
to the current stage, i.e., the stage just before selecting $c$.

\begin{definition}\label{def: RES}
The \defstyle{loop avoidance protocol} (\LA) asserts 
that in every round, every player $i$
should avoid---if possible---all choices $c$ for which the
following condition holds:
if the player $i$ selects $c$, then there exist
choices for the other players so that the 
resulting stage is automorphism-equivalent to
the current stage.
If this condition holds for all choices of the player $i$, then $i$ makes a random choice. Moreover, uniform probability is
used among all the possible choices of $i$.
%
\end{definition}


It is easy to see that \LA avoids, when possible, 
all such stages that are structurally similar to 
\emph{any} earlier stage in the repeated play.
As structurally similar stages are
essentially identical (cf. Lemma \ref{similarity lemma}),
repetition of such stages can be
seen as a ``loop'' in the repeated play. 
When trying to \emph{guarantee} coordination as
quickly as possible, such loops should be avoided.
%
%
%
%
In addition to this heuristic justification,
Theorems \ref{the: even m} and \ref{the: odd m}
give a fully compelling justification for \LA
when considering guaranteed coordination in two-player \CM-games.
For now, we present the following propositions
(see Appendix~\ref{appendix: Proofs} for proofs); 
see also Example~\ref{ex: use of LA} in Appendix~\ref{appendix: Examples}
for an illustration of the use of \LA. 
%

\begin{proposition}\label{the: LG3}
$\LA$ is the uniquely \ECT-optimal and uniquely \GCT-optimal in $\CM_3$.
\end{proposition}

%

\begin{proposition}\label{the: GCT with RES}
$\LA$ guarantees coordination in games $\CM_m$ in $\lceil m/2 \rceil$ rounds when $m$ is odd, 
but $\LA$ does not guarantee coordination in $\CM_m$ for any even $m$.
\end{proposition}



We next present the 
\emph{wait-or-move protocol} \WM, which naturally
appears in numerous real-life two-player coordination scenarios. 
Informally, both players
alternate (with equal probability) between two choices:
the players own initial choice and another choice that coordinates 
with the initial choice of the other
player. 
%
%
%

\begin{definition}\label{definition: waitormove}
The \defstyle{wait-or-move protocol (\WM)} for repeated two-player \WLC-games 
goes as follows: first select randomly any choice $c$, and 
thereafter choose with equal probability $c$ or a choice $c'$ that
coordinates with the initial choice of the other player (thereby never
picking other choices than $c$ and $c'$). Definition
\ref{wmseconddef} in Appendix~\ref{appendix: Examples} specifies
\WM in more detail.
\end{definition}


The following theorem shows that $\WM$ is very fast in 
relation to {\ECT}s.
This holds for \emph{all} two-player \WLC-games, not only
choice matching games $\CM_m$.
The proof is given in Appendix~\ref{appendix: Proofs}.

\begin{theorem}\label{jormatheorem}
Let $G$ be a \WLC-game with
one-shot coordination probability $p$
when both players make their first choice randomly.
Then the expected coordination time by \WM is at most $3-2p$. 
\end{theorem}

\begin{corollary}\label{jormacorollary}
The \ECT with \WM is strictly less than $3$ in
every two-player \WLC-game.
\end{corollary}

It follows from the proof of Theorem \ref{jormatheorem} that 
the \ECT with \WM is \emph{exactly} $3-\frac{2}{m}$ in all
choice matching games $\CM_m$. 
Thus Corollary \ref{jormacorollary} cannot be improved, as
the {\ECT}s of the games $\CM_m$
grow asymptotically closer to the strict upper bound $3$ 
when $m$ is increased. 
%
%
%
In the particular case of $\CM_2$, the \ECT with \WM is $3-\frac{2}{2}=2$.
Thus the following lemma clearly holds.

\begin{lemma}\label{roskalemma}
When $S =(\CM_m,\mathcal{H}_k)$ is a non-final stage with
exactly two touched edges, then the \ECT from $S$ with \WM is exactly 2.
Moreover, in any \WLC-game $G$, if $S' = (G,\mathcal{H}_k)$ is a non-final stage 
that is reacbable by using \WM, then the \ECT from $S'$ with \WM is
at most 2.
\end{lemma}

\vcut{
\begin{proof}
In every round, 
the one-shot coordination probability is $\dfrac{1}{2}$.
Thus the probability of
coordinating in the $m$th round (and not earlier) is
%
%
$(\frac{1}{2})^{m-1} \cdot \frac{1}{2} = (\frac{1}{2})^m$
%
%
for all $m\geq 1$. 
Hence $ECT$ is
%
%
$E\ =\ \sum\limits_{\scriptscriptstyle k\geq 1}^{\scriptscriptstyle\infty} \dfrac{k}{2^{k}}$
%
%
and it is well known that the right hand side sum equals 2.
\end{proof}
}

\WM eventually leads to coordination with asymptotic
probability $1$ in all two-player \WLC-games. Nevertheless, 
it clearly does not
guarantee (with certainty) coordination in any number of rounds
in \WLC-games where the winning relation is not the total relation.
%
%
In a typical real-life scenario, eternal non-coordination is of
course impossible by \WM, but it is conceivable, for
example, that two computing units using the very same pseudorandom number 
generator will never coordinate due to being synchronized to swap their choices in
precisely the same rounds.

It is easy to show that \WM is the unique protocol which
gives the optimal \ECT (namely, 2 rounds) in the ``droitwich-scenario'' 
of the game $\CM_2$ (see Appendix~\ref{appendix: Proofs} for a proof):

\begin{proposition}\label{the: LG2}
\WM is uniquely \ECT-optimal in $\CM_2$.
\end{proposition}


Next we compare
compare the pros and cons of \LA and \WM in
two-player choice matching games $\CM_m$.
Recall that \WM does not guarantee coordination in these games (when $m\not=1$),
while \LA does guarantee coordination in $\CM_m$ if and only if $m$ is \emph{odd}.
Concerning \emph{expected} coordination times, it is easy to prove that \WM
gives a smaller \ECT than \LA in $\CM_m$ for all
even $m$ (except for the case $m=2$, where \WM and \LA behave identically).
Thus we now restrict attention to the games $\CM_m$ with odd $m$.
Then, the probability of coordinating in the $\ell$-th
round of $\CM_m$ using \LA, with $\ell\leq \lceil m/2 \rceil$, can
relatively easily be seen to be calculable by 
the formula $P_{\ell,m}$  
defined below (where the product is 
$1$ when $\ell=1$). And using the formula for $P_{\ell,m}$, we also 
get a formula for the expected coordination 
time $E_m$ in $\CM_m$ with $\LA$:
\[
P_{\ell,m} = \frac{1}{m - 2\ell + 2}\prod\limits_{k=0}^{k\, =\, {\ell - 2}}\frac{m - 2k - 1}{m - 2k},
\qquad
	E_m =\!\!\! \sum\limits_{\ell=1}^{\ell=\lceil m/2\rceil}\!\!\!\ell\cdot P_{\ell,m}.
\]
%
%
%
%
%
%
Using this and Theorem \ref{jormatheorem}, we 
can compare the {\ECT}s in $\CM_m$ with \LA and \WM for odd~$m$.
%


\begin{center}
\scalebox{0.75}{
\begin{tabular}{|c|c|c|}
\hline
$m$ & \ECT in $\CM_m$ with \WM & \ECT in $\CM_m$ with \LA \\
\hline
$1$ & 1 & $1$ \\
$3$ & $2+\frac{1}{3}$ & $1+\frac{2}{3}$ \\[1mm]
$5$ & $2+\frac{3}{5}$ & $2+\frac{1}{3}$ \\[1mm]
$7$ & $2+\frac{5}{7}$ & $3$ \\[1mm]
$9$ & $2+\frac{7}{9}$ & $3+\frac{2}{3}$ \\[1mm]
\hline
\end{tabular}
}
\end{center}


\vcut{
We have $E_1 = 1$, 
%
$E_3 = 1\frac{2}{3}$,
%
$E_5 = 2\frac{1}{3}$,
$E_7 = 3$ 
and 
$E_9 = 3\frac{2}{3}$. 
}

Especially the case $m=7$ is
interesting, as the \ECT with \LA is exactly 3 which is precisely the
strict upper bound for the {\ECT}s with \WM 
for the class of all two-player choice matching games $\CM_m$.
Furthermore, $m=7$ is the case 
where \WM becomes
faster than \LA in
relation to {\ECT}s. Thus \WM clearly
stays faster than \LA for all $m\geq 7$, including even values of $m$.

\section{Optimizing guaranteed coordination times}\label{sec: GCTs}

In this section we investigate when coordination can be guaranteed in two-player \CM-games and which protocols give the optimal \GCT for them. We begin with the following result.

\begin{theorem}\label{the: even m}
For all even $m\geq 2$, there is no protocol which
guarantees coordination in $\CM_m$.
\end{theorem}

\begin{proof}
Let $\pi$ be a protocol. As $\pi$ is structural, it is possible that in
each round of $\CM_m$, the players pick a pair $(c,c')$ of choices
that are structurally equivalent. Suppose this indeed happens. Now, in each 
round, there are two types of choices the players can make: (1) they both pick a choice from a touched edge; or (2) they both pick a choice from an untouched edge. As there is always an even number of untouched edges left in the game, the
choice of type (2) will never guarantee coordination.
And when the players have failed to
coordinate so far, they will never succeed by making a choice of type (1) (due to
structural equivalence of the choices).
\end{proof}

We next consider choice matching games $\CM_m$
with an odd $m$. Proposition~\ref{the: GCT with RES} showed that the \GCT with \LA
in these games is $\lceil m/2 \rceil$. 
The next theorem (proved in Appendix~\ref{appendix: Proofs})
shows that this is the optimal \GCT for $\CM_m$, and
moreover, \LA is the unique protocol giving this \GCT. 

\begin{theorem}\label{the: odd m}
For any odd $m\geq 1$, $\LA$ is uniquely \GCT-optimal for $\CM_m$.
\end{theorem}

\section{Optimizing expected coordination times}\label{optimizingsection}

In this section we investigate which protocols give the best {\ECT}s 
for two-player choice matching games. 
We also 
investigate when the best \ECT is obtained by a unique protocol.
We already know by Propositions~\ref{the: LG2} and \ref{the: LG3} that the
optimal {\ECT}s for $\CM_2$ and $\CM_3$
are uniquely given by \WM and \LA, respectively. 
Thus it remains to consider the games $\CM_m$ with $m\geq 4$.
We first cover
the case $m\geq 6$ and show that then \WM is the unique protocol
giving the best \ECT. 
The remaining special cases $m=4$ and $m=5$ will then be examined.
The following auxiliary lemma (proven in Appendix~\ref{appendix: Proofs}) will be
used in the proofs.

\begin{lemma}\label{32 Lemma}
The \ECT from $(\CM_m,\mathcal{H}_k)$ with no focal point 
is at least $\frac{3}{2}$ with any protocol.
\end{lemma}

We then present a formula for
estimating the best {\ECT}s in cases to be investigated.
Let $S := (\CM_m,\mathcal{H}_k)$ be a non-final stage
with exactly two touched edges.
Thus there are $n:=m-2$ untouched edges.
Suppose the players use a protocol $\pi$ behaving as follows in
round $k+1$. Both players pick a choice from
some touched edge with probability $p$ and
from an untouched edge with probability $(1-p)$. A uniform
distribution is used on choices in 
both classes: probability $\frac{p}{2}$ for both
choices on touched edges (which makes
sense by Lemma \ref{Two-edge lemma}) and probability $\frac{1-p}{n}$ for each
choice on untouched edges (which is necessary with a 
structural protocol). If one player selects a choice $c$ from a touched
edge and the other one a choice $c'$ from an untouched
edge, the players win in the next round by choosing the 
edge with~$c'$.
%
%
Note that $c'$ is a focal point, so the winning edge can be
chosen by a structural protocol with probability~1.
(Also other focal points arise which could alternatively be
used; cf. Example~\ref{ex: focal points} in Appendix~\ref{appendix: Examples}.)


Suppose then that $E_1$ is the \ECT with $\pi$ from a stage $(\CM_m,\mathcal{H}_{k+1})$
where both players have chosen a touched
edge in round $k+1$ but failed to coordinate.
Two different such stages $(\CM_m,\mathcal{H}_{k+1})$ exist, but they
are automorphism-equivalent, so $\pi$ can give the same \ECT 
from both of them by Lemma~\ref{similarity lemma}. (Indeed, if $\pi$ gave
two different $\ECT$s, it would make sense to adjust it to give the
smaller one.)
Similarly, suppose $E_2$ is the \ECT with $\pi$ from a
stage $(\CM_m,\mathcal{H}_{k+1}')$ where
both players have chosen an untouched edge in round $k+1$ but failed to coordinate.
Note that all possible such stages $(\CM_m,\mathcal{H}_{k+1}')$
are renamings of each other, so $\pi$ must give the same \ECT from each one.
We next establish that the expected
coordination time from $(\CM_m,\mathcal{H}_k)$ with $\pi$ is now
given by the following formula (to be called formula (E) below):
\begin{align*}
	p^2\Bigl(\frac{1}{2}+\frac{1}{2}\bigl(1+E_1\bigr)\Bigr)\ &+\ 2p(1-p)\cdot 2\ 
	+\ (1-p)^2\, \Bigl(\dfrac{1}{n} + \dfrac{n-1}{n}\Bigl(1 + E_2\Bigr)\Bigr) \tag{E}
\end{align*}

Indeed, both players choose a touched edge in round $k+1$ with probability $p^2$. 
In that case the \ECT from $(\CM_m,\mathcal{H}_{k})$ is 
$\frac{1}{2}+\frac{1}{2}(1+E_1)$,
the first occurrence of $\frac{1}{2}$ corresponding to direct
coordination and the remaining term covering the case where
coordination fails at first.
Both players choose an untouched edge in round $k+1$ with
probability $(1-p)^2$, and then
the \ECT from $(\CM_m,\mathcal{H}_k)$ is
$\frac{1}{n} + \frac{n-1}{n}(1 + E_2).$
\vcut{
Here $\frac{1}{n}$ gives the contribution resulting from direct
coordination and the remaining term covers the case if
coordination fails at first. 
}
The remaining term $2p(1-p)\cdot 2$ is the
contribution of the case
where one player chooses a touched edge and the other player an
untouched one. The probability for this is $2p(1-p)$, and
the remaining factor $2$ indicates that coordination 
immediately happens in the subsequent round $k+2$ using the
focal point created in round $k+1$.

Now consider the following informal argument sketch.
In $\CM_m$ with $m\geq 6$, we may assume that $E_1\leq 2$
and $E_2\geq\frac{3}{2}$ by Lemmas \ref{roskalemma} and \ref{32 Lemma}.
Figure \ref{fig: fig} below illustrates the graph of (E) 
with $E_1 = 2$, $E_2 = \frac{3}{2}$, $n=4$, so then (E) has a unique
minimum at $p=1$ when $p\in[0,1]$. This
suggests that---under these parameter values---the players should
always choose a touched edge in
stages with exactly two touched edges. Clearly, lowering $E_1$,
raising $E_2$ or raising $n$ should make it even more beneficial to
choose a touched edge. As we indeed can assume that $E_1\leq 2$
and $E_2\geq\frac{3}{2}$ in $\CM_m$ for $m\geq 6$, this informally 
justifies that the following theorem holds.

\begin{theorem}\label{the: 6-ladders}
\WM is uniquely \ECT-optimal for each $\CM_m$ with $m\geq 6$.
\end{theorem}

\begin{proof}
Let $S := (\CM_m,\mathcal{H}_k)$, $m\geq 6$, be a non-final stage with
precisely two touched edges and $S'$ a stage extending $S$ by one
round where the players both choose an untouched edge but fail to coordinate.
Let $r_1$ (respectively, $r_2$) be the infimum of all
possible {\ECT}s from $S$ (respectively, $S'$) with different protocols.
Note that by Lemma \ref{similarity lemma}, $r_1$ and $r_2$ are independent of
which particular representative stages we
choose, as long as the stages satisfy the given constraints.
Let $\epsilon > 0$ and fix some numbers $E_1$ and $E_2$
such that $|E_1 - r_1 | < \epsilon$ and $|E_2 - r_2 | < \epsilon$. We
assume $E_1\leq 2$
and $E_2\geq\frac{3}{2}$ by Lemmas \ref{roskalemma} and \ref{32 Lemma}.
It is easy to show that with such $E_1$
and $E_2$, the minimum value of the formula (E)
with $p\in [0,1]$ is obtained at $p = 1$ (for any $n = m-2 \geq 4$).

Thus, after the necessarily random choice in round one, the above reasoning
shows that the players should choose a
touched edge with probability $p=1$ in
each round. Indeed, assume the earliest occasion that some
protocol $\pi_k$ assigns $p\not=1$ in some stage is round $k$.
Then the above shows that the \ECT of $\pi_k$
can be strictly
improved by letting $p=1$ in that round. By Lemma 
\ref{Two-edge lemma} in the Appendix, a uniform probability over the
touched choices should be used. 
\end{proof}

\vspace{-3.3mm}

\setlength{\belowcaptionskip}{-15pt}
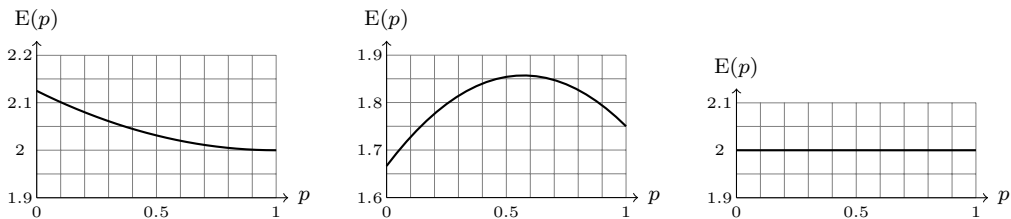
\begin{figure}[h]
\begin{center}
\begin{tikzpicture}[domain=0:1,xscale=5,yscale=10,scale=0.63]
 	\draw[ultra thin,color=gray] (0,1.9) grid[xstep=0.1,ystep=0.05] (1,2.2);
	\draw[->] (0,1.9) -- (1.05,1.9) node[right] {\scriptsize $p$};
	\draw[->] (0,1.9) -- (0,2.23) node[above] {\scriptsize E$(p)$};
	\node at (0,1.87) {\tiny $0$};
	\node at (0.5,1.87) {\tiny $0.5$};
	\node at (1,1.87) {\tiny $1$};
	\node at (-0.07,1.9) {\tiny $1.9$};
	\node at (-0.07,2) {\tiny $2$};
	\node at (-0.07,2.1) {\tiny $2.1$};
	\node at (-0.07,2.2) {\tiny $2.2$};
	\draw[thick] plot (\x, {\x^2*(1/2+(1/2)*(1+2)) + 4*\x*(1-\x) + (1-\x)^2*(1/4+(3/4)*(1+(3/2)))});
\end{tikzpicture}
\;
\begin{tikzpicture}[domain=0:1,xscale=5,yscale=10,scale=0.63]
 	\draw[ultra thin,color=gray] (0,1.6) grid[xstep=0.1,ystep=0.05] (1,1.9);
	\draw[->] (0,1.6) -- (1.05,1.6) node[right] {\scriptsize $p$};
	\draw[->] (0,1.6) -- (0,1.93) node[above] {\scriptsize E$(p)$};
	\node at (0,1.57) {\tiny $0$};
	\node at (0.5,1.57) {\tiny $0.5$};
	\node at (1,1.57) {\tiny $1$};
	\node at (-0.07,1.6) {\tiny $1.6$};
	\node at (-0.07,1.7) {\tiny $1.7$};
	\node at (-0.07,1.8) {\tiny $1.8$};
	\node at (-0.07,1.9) {\tiny $1.9$};
	\draw[thick] plot (\x, {\x^2*(1/2+(1/2)*(1+(3/2))) + 4*\x*(1-\x) + (1-\x)^2*(1/3+(2/3)*(1+1))});
\end{tikzpicture}
\;
\begin{tikzpicture}[domain=0:1,xscale=5,yscale=10,scale=0.63]
 	\draw[ultra thin,color=gray] (0,1.9) grid[xstep=0.1,ystep=0.05] (1,2.1);
	\draw[->] (0,1.9) -- (1.05,1.9) node[right] {\scriptsize $p$};
	\draw[->] (0,1.9) -- (0,2.13) node[above] {\scriptsize E$(p)$};
	\node at (0,1.87) {\tiny $0$};
	\node at (0.5,1.87) {\tiny $0.5$};
	\node at (1,1.87) {\tiny $1$};
	\node at (-0.07,1.9) {\tiny $1.9$};
	\node at (-0.07,2) {\tiny $2$};
	\node at (-0.07,2.1) {\tiny $2.1$};
	\draw[thick] plot (\x, {\x^2*(1/2 + (1/2)*(1+2)) + 4*\x*(1-\x) + (1-\x)^2*(1/2+(1/2)*(1+2))});
\end{tikzpicture}%
\caption{Graph of (E) with\ \ (i)~$n=4$, $E_1=2$, $E_2=\frac{3}{2}$;\ \ \ 
(ii)~$n=3$, $E_1=\frac{3}{2}$, $E_2=1$;\ \ \ 
(iii)~$n=2$, $E_1=E_2=2$.}
\label{fig: fig}
\end{center}
\end{figure}%
%
%

We then cover the case for $\CM_5$. The argument is similar to 
the case for $\CM_m$ with $m\geq 6$, 
but this time leads to the use of $\LA$ instead of $\WM$.

\begin{theorem}\label{the: 5-ladders}
For $\CM_5$, \LA is uniquely \ECT-optimal.
\end{theorem}

\begin{proof}
Let $S := (\CM_5,\mathcal{H}_k)$ be a non-final stage with
precisely two touched edges and $S'$ a stage extending $S$ by one
round where the players both choose an untouched edge but fail to coordinate. The \ECT-optimal protocol from $S'$
chooses the unique winning pair of focal points in round $k+2$, so we
now have $E_2 = 1$.
Let $r_1$ be the infimum of all
possible {\ECT}s from $S$ with different protocols.
Let $\epsilon > 0$ and fix some real number $E_1$
such that $|E_1 - r_1 | < \epsilon$,
assuming $E_1\geq\frac{3}{2}$ (cf. Lemma \ref{32 Lemma}).
It is straightforward to show that with these values,
and with $n=3$, the minimum of (E)
when $p\in [0,1]$ is obtained at $p = 0$. (See
also Figure~\ref{fig: fig} for the graph of (E) when $E_1=\frac{3}{2}$ 
for an illustration. Even then the figure suggests to choose an
untouched edge.)

Thus, after the necessarily random choice in round one, the above reasoning
shows that the players should choose an
untouched edge with probability $1$ in
the second round, thereby following $\LA$.
Coordination is guaranteed (latest) in the third round.
%
\end{proof}


In the last case $m=4$, \WM is \ECT-optimal, but not uniquely, as there
exist infinitely many other \ECT-optimal protocols. 
The reason for this is that---as shown in Figure \ref{fig: fig}---the graph of (E)
becomes the constant line with the value $2$ in
special case where $E_1=E_2=2$, and then any $p\in[0,1]$
gives the optimal value for (E). 
%
%
A complete proof is given in Appendix~\ref{appendix: Proofs}.


\begin{theorem}\label{the: 4-ladders}
\WM is \ECT-optimal for $\CM_4$, but there are continuum many other
protocols that are also \ECT-optimal. 
\end{theorem}

We have now given a complete analysis of optimal {\ECT}s and  {\GCT}s in
two-player \CM-games summarized in Figure~\ref{fig: summary table}.
See Appendix~\ref{appendix: remarks} for
further discussion on optimal play in \CM-games.

\section{The hardest two-player \WLC-games}\label{sec: characterisations}

In this section we give an optimal characterization of the upper 
bounds of {\ECT}s in \WLC-games as a
function of game size. 
For any $m\geq 1$, an \defstyle{$m$-choice game}
refers to any \emph{two-player} \WLC-game $G=(A,C_1,C_2,W_G)$ where $m=\max\{|C_1|,|C_2|\}$. 
Note that, with the classical matrix representation of an $m$-choice game, the parameter $m$ corresponds to the largest dimension of the matrix.
In this section we will also show 
that $\CM_m$ can be seen as the
hardest $m$-choice game for all $m\neq 3$ (see Corollary~\ref{hardest games}).

Our first theorem shows that the wait-or-move protocol is reasonably ``safe'' to use in 
any $m$-choice game with $m\not\in\{3,5\}$ as it
always guarantees an \ECT which is at most equal
to the upper bound of optimal {\ECT}s of
all $m$-choice games for the particular $m$.

\begin{theorem}\label{safety of WM}
Let $m\not\in\{1,3,5\}$ and consider an
$m$-choice game $G=(A,C_1,C_2,W_G)\not= \CM_m$.
Then the \ECT in $G$ with \WM is strictly smaller than the
optimal \ECT in $\CM_m$.
\end{theorem}

\begin{proof}
By Theorems \ref{the: 6-ladders}, \ref{the: 4-ladders} and Proposition~\ref{the: LG2},
the optimal \ECT in $\CM_m$ is given by \WM. 
We saw in Section \ref{sect:solvability}
that the {\ECT} with \WM is $3-\frac{2}{m}$ in $\CM_m$ 
and at most $3-2p$ in $G$, where $p$ is
the one-shot coordination probability when choosing randomly in $G$.  
Since $G$ is an $m$-choice game, $|W_G|\geq m$. 
If $|W_G|>m$, then $p > \frac{m}{m^2} = \frac{1}{m}$.
And if $|W_G|=m$, we have $p = \frac{m}{mn} = \frac{1}{n} > \frac{1}{m}$
where $n:=\min\{|C_1|,|C_2|\}<m$ since $G\neq\CM_m$.
In both cases, we have $3-2p < 3-\frac{2}{m}$.
\end{proof}


By the \defstyle{greatest optimal \ECT} among a class $\mathcal{G}$ of \WLC-games, we mean a value $r$ such that (1) $r$ is the optimal \ECT for some $G\in\mathcal{G}$; and (2) for every $G\in\mathcal{G}$, there is a protocol
which gives it an $\ECT\leq r$. 
By Theorem~\ref{safety of WM}, the greatest optimal \ECT among $m$-choice
games is given by \WM in $\CM_m$ for $m\not\in\{1,3,5\}$.
Also the special cases of $1$, $3$ and $5$ are covered below:

\begin{theorem}\label{the: upper bounds for ECTs}
For any $m$, the greatest optimal \ECT among $m$-choice games is given below:

\begin{center}
\scalebox{0.8}{
\begin{tabular}{|c|c|c|c|}
\hline Game size
& $m\in\mathbb{Z}_+\setminus\{3,5\}$ & \quad$m=5$\quad\text{} & $m=3$ \\
\hline \rule{0pt}{1.2\normalbaselineskip} 
Greatest optimal\; \ECT & $3-\frac{2}{m}$ & $2+\frac{1}{3}$ & $\frac{1+\sqrt{4+\sqrt{17}}}{2}$ $(\approx 1,925)$ \\[1mm]
\hline
\end{tabular}
}
\end{center}

\end{theorem}

\begin{proof}
The case $m=1$ is trivial and the cases $m\not\in\{3,5\}$
follow from Theorem 
\ref{safety of WM}.
When $m=3$ or $m=5$, we need to systematically
cover all $m$-choice games and give
estimates for {\ECT}s in them. 
This is done in Appendix~\ref{appendix: 3 & 5},
where we provide an extensive graph theoretic
analysis of all $3$-choice and $5$-choice games.
It turns out that the greatest optimal \ECT among 5-choice
games is realized in $\CM_5$ (and no other 5-choice game).
%
%
For $m=3$, the greatest optimal \ECT is also
realized by a single \WLC-game. This game is pictured below.
\end{proof}


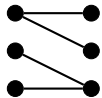
\begin{center}
\begin{tikzpicture}[scale=0.5,choice/.style={draw, circle, fill=black!100, inner sep=2.1pt}]
	\node at (0,0) [choice] (00) {};
	\node at (2,0) [choice] (20) {};
	\node at (0,1) [choice] (01) {};
	\node at (2,1) [choice] (21) {};
	\node at (0,2) [choice] (02) {};
	\node at (2,2) [choice] (22) {};
	\draw[thick] (00) to (20);
	\draw[thick] (01) to (20);
	\draw[thick] (02) to (21);
	\draw[thick] (02) to (22);
\end{tikzpicture}
\end{center}


As the greatest optimal \ECT is
realized uniquely by $\CM_5$, the following holds by Theorem~\ref{safety of WM}:

\begin{corollary}\label{hardest games}
For $m\not= 3$, the greatest optimal \ECT among m-choice games is
uniquely realized by $\CM_m$.
\end{corollary}

Hence choice matching games can indeed be seen as \emph{the most difficult} two-player \WLC-games---excluding the interesting special case of 3-choice games as
discussed above.
%
%

\section{Conclusion}


In this paper we gave a complete analysis for two-player \CM-games with respect to both {\GCT}s and {\ECT}s. We also found optimal upper bounds for optimal {\ECT}s for all two-player \WLC-games when determined according to game size only. 
A highly challenging
next step would be to find complete characterizations for optimal $\ECT$s
(and $\GCT$s) for all \WLC-games when determined by the
full structure of the game.

\newpage

\appendix

\section{Appendix: Examples and extra definitions}\label{appendix: Examples}

\begin{example}\label{drawingexamples}
\label{ex: game graph} Here we give two examples of drawings of \WLC-games: a two-player game $G_1$ with $3$ choices for both players and a total of $6$ winning profiles repesented as edges; and a three-player \WLC-game $G_2$ with $2$ choices for each player and $4$ winning profiles, each represented as a triple of choices connected by (solid or dotted) lines. 
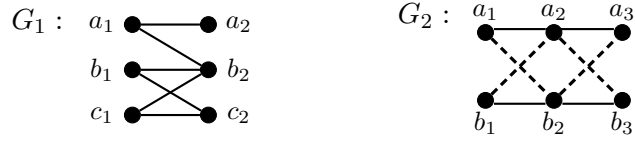
\begin{center}
\begin{tikzpicture}[scale=0.5,choice/.style={draw, circle, fill=black!100, inner sep=2.1pt}]
	\node at (-2.5,2.5) {$G_1:$};	
	\node at (2,0) [choice] (20) {};
	\node at (2,1.2) [choice] (21) {};
	\node at (0,1.2) [choice] (01) {};
	\node at (2,2.4) [choice] (22) {};
	\node at (0,0) [choice] (00) {};
	\node at (0,2.4) [choice] (02) {};	
	\node at (0,-0.5) {};
	\draw[thick] (00) to (20);	
	\draw[thick] (00) to (21);
	\draw[thick] (01) to (20);	
	\draw[thick] (01) to (21);
	\draw[thick] (02) to (21);
	\draw[thick] (02) to (22);
	\node at (-0.8,2.4) {\small $a_1$};
	\node at (2.8,2.4) {\small $a_2$};
	\node at (-0.8,1.2) {\small $b_1$};
	\node at (-0.8,0) {\small $c_1$};
	\node at (2.8,0) {\small $c_2$};	
	\node at (2.8,1.2) {\small $b_2$};	
\end{tikzpicture}
\qquad \qquad 
\begin{tikzpicture}[scale=0.9,choice/.style={draw, circle, fill=black!100, inner sep=2.1pt}]
	\node at (-0.9,1.25) {$G_2:$};
	\node at (0,0) [choice] (00) {};
	\node at (1,0) [choice] (10) {};
	\node at (2,0) [choice] (20) {};
	\node at (0,1) [choice] (01) {};
	\node at (1,1) [choice] (11) {};
	\node at (2,1) [choice] (21) {};
	\draw[thick] ([yshift=-0.5mm]00.east) to ([yshift=-0.5mm]10.west);
	\draw[very thick,densely dashed] (00) to (11);
	\draw[thick] ([yshift=0.5mm]01.east) to ([yshift=0.5mm]11.west);
	\draw[thick] ([yshift=-0.5mm]10.east) to ([yshift=-0.5mm]20.west);
	\draw[very thick,densely dashed] (10) to (21);
	\draw[thick] ([yshift=0.5mm]11.west) to ([yshift=0.5mm]21.west);
	\draw[very thick,densely dashed] (01) to (10);
	\draw[very thick,densely dashed] (11) to (20);
	\node at (0,1.3) {\small $a_1$};
	\node at (1,1.3) {\small $a_2$};
	\node at (2,1.3) {\small $a_3$};
	\node at (0,-0.33) {\small $b_1$};
	\node at (1,-0.33) {\small $b_2$};
	\node at (2,-0.33) {\small $b_3$};
\end{tikzpicture}
\end{center}
\vspace{-3mm}
\end{example}

We now specify some useful notational conventions from \cite{lori2017} for
identifying some special \WLC-games (\emph{see also the figure below for
related examples}).
\begin{itemize}[leftmargin=*]
\item Let $m_1,\dots, m_n\in\mathbb{Z}_+$. We write $G(m_1\times \cdots \times m_n)$ for the $n$-player \WLC-game where the player $i$ has $m_i$ choices and the winning relation is the \emph{universal relation} $C_1\times \cdots \times C_n$.

\item Let $m\geq 2$. We write $G(O_m)$ for the two-player \WLC-game in which both players have $m$ choices and the winning relation $W_G$
forms a $2m$-cycle through all the $2m$ choices. (Thus the game graph of this $\WLC$-game
corresponds to the cycle graph $C_{2m}$.)
Similarly we write $G(Z_m)$ for the two-player \WLC-game where both players have $m$ choices and $W_G$ forms a $(2m-1)$-edge path through all choices. Moreover $G(\Sigma_m)$ denotes a $\WLC$-game where the player~1 has $m-1$ choices, the player 2 has $m$ choices and $W_G$ forms a $(2m-2)$-edge path through all the choices; the game obtained by permuting the players in $G(\Sigma_m)$ is denoted by $G(\reflectbox{$\Sigma$}_m)$.

\item Suppose that $G(A)$ and $G(B)$ have been defined and both have the same number of players.
Then $G(A + B)$ is the \emph{disjoint union} of $G(A)$ and $G(B)$, i.e.,
the game obtained by assigning to each player a disjoint union of her/his choices in
 $G(A)$ and $G(B)$, with the winning relation for $G(A + B)$
being the union of the winning relations in $G(A)$ and $G(B)$. 


\item If $m\in\mathbb{Z}_+$, then  $G(mA):=G(A + \cdots + A)$ (with $A$ repeated $m$ times). 
\end{itemize}

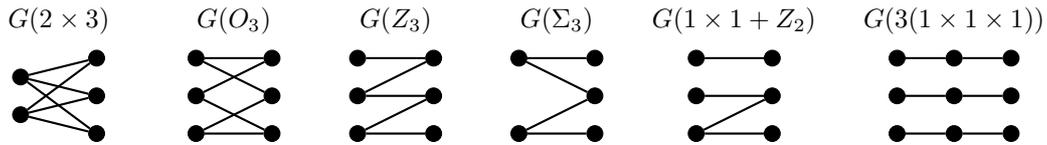
\begin{center}\label{notationalconventionsfigure}
\begin{tikzpicture}[scale=0.5,choice/.style={draw, circle, fill=black!100, inner sep=2.1pt}]
	\node at (1,3) {\small\bf $G(2\times 3)$};	
	\node at (0,0.5) [choice] (00) {};
	\node at (0,1.5) [choice] (01) {};
	\node at (2,0) [choice] (20) {};
	\node at (2,1) [choice] (21) {};
	\node at (2,2) [choice] (22) {};
	\draw[thick] (00) to (20);	
	\draw[thick] (00) to (21);
	\draw[thick] (00) to (22);
	\draw[thick] (01) to (20);	
	\draw[thick] (01) to (21);
	\draw[thick] (01) to (22);
\end{tikzpicture}
\qquad
\begin{tikzpicture}[scale=0.5,choice/.style={draw, circle, fill=black!100, inner sep=2.1pt}]
	\node at (1,3) {\small\bf $G(O_3)$};	
	\node at (0,0) [choice] (00) {};
	\node at (2,0) [choice] (20) {};
	\node at (0,1) [choice] (01) {};	
	\node at (2,1) [choice] (21) {};
	\node at (0,2) [choice] (02) {};	
	\node at (2,2) [choice] (22) {};
	\draw[thick] (00) to (20);	
	\draw[thick] (00) to (21);
	\draw[thick] (02) to (22);
	\draw[thick] (01) to (22);
	\draw[thick] (02) to (21);
	\draw[thick] (01) to (20);
\end{tikzpicture}
\qquad
\begin{tikzpicture}[scale=0.5,choice/.style={draw, circle, fill=black!100, inner sep=2.1pt}]
	\node at (1,3) {\small\bf $G(Z_3)$};	
	\node at (0,0) [choice] (00) {};
	\node at (2,0) [choice] (20) {};
	\node at (0,1) [choice] (01) {};	
	\node at (2,1) [choice] (21) {};
	\node at (0,2) [choice] (02) {};	
	\node at (2,2) [choice] (22) {};
	\draw[thick] (00) to (20);	
	\draw[thick] (00) to (21);
	\draw[thick] (01) to (21);
	\draw[thick] (01) to (22);
	\draw[thick] (02) to (22);	
\end{tikzpicture}
\qquad
\begin{tikzpicture}[scale=0.5,choice/.style={draw, circle, fill=black!100, inner sep=2.1pt}]
	\node at (1,3) {\small\bf $G(\Sigma_3)$};	
	\node at (0,0) [choice] (00) {};
	\node at (2,0) [choice] (20) {};
	\node at (2,1) [choice] (21) {};
	\node at (0,2) [choice] (02) {};	
	\node at (2,2) [choice] (22) {};
	\draw[thick] (00) to (20);	
	\draw[thick] (00) to (21);
	\draw[thick] (02) to (21);
	\draw[thick] (02) to (22);	
\end{tikzpicture}
\quad
\begin{tikzpicture}[scale=0.5,choice/.style={draw, circle, fill=black!100, inner sep=2.1pt}]
	\node at (1,3) {\small\bf $G(1\times 1 + Z_2)$};	
	\node at (0,0) [choice] (00) {};
	\node at (2,0) [choice] (20) {};
	\node at (0,1) [choice] (01) {};	
	\node at (2,1) [choice] (21) {};
	\node at (0,2) [choice] (02) {};	
	\node at (2,2) [choice] (22) {};
	\draw[thick] (01) to (21);
	\draw[thick] (00) to (20);	
	\draw[thick] (00) to (21);
	\draw[thick] (02) to (22);	
\end{tikzpicture}
\;
\begin{tikzpicture}[scale=0.5,choice/.style={draw, circle, fill=black!100, inner sep=2.1pt}]
	\node at (1.5,3) {\small\bf $G(3(1\times 1\times 1))$};	
	\node at (0,0) [choice] (1c) {};
	\node at (1.5,0) [choice] (2c) {};
	\node at (3,0) [choice] (3c) {};
	\node at (0,1) [choice] (1b) {};	
	\node at (1.5,1) [choice] (2b) {};
	\node at (3,1) [choice] (3b) {};
	\node at (0,2) [choice] (1a) {};	
	\node at (1.5,2) [choice] (2a) {};
	\node at (3,2) [choice] (3a) {};
	\draw[thick] (1a) to (2a);	
	\draw[thick] (2a) to (3a);
	\draw[thick] (1b) to (2b);	
	\draw[thick] (2b) to (3b);
	\draw[thick] (1c) to (2c);	
	\draw[thick] (2c) to (3c);
\end{tikzpicture}
\end{center}

Note that the game $G(m(1\times 1))$ is the two-player choice matching game $\CM_m$.

\begin{example}\label{ex: Renaming}
Below we have two stages $(G,\mathcal{H}_2)$ and $(G',\mathcal{H}_2')$,
where the players have selected the choices with dotted circles in round 1
and the choices with solid circles in round 2.
There is a renaming between the
stages $(G,\mathcal{H}_2)$ and $(G',\mathcal{H}_2')$.
This is because if we first swap the players in $(G,\mathcal{H}_2)$,
then there will be an isomorphism to $(G',\mathcal{H}_2')$. Also note that the choices $c$ and $d$ are structurally equivalent in the initial stage $(G,\mathcal{H}_0)$, but this equivalence is broken when the player 2 selects $c$ in the first round.
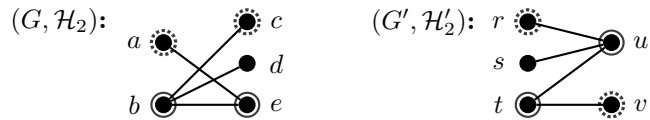
\begin{center}
\begin{tikzpicture}[scale=0.55,choice/.style={draw, circle, fill=black!100, inner sep=2.1pt},
	location1/.style={draw, very thick, densely dotted, circle, color=black!77, inner sep=3.2pt},
	location2/.style={draw, thick, circle, color=black!77, inner sep=3.2pt}]
	\node at (-2.5,2) {\small\bf $(G,\mathcal{H}_2)$:};		
	\node at (0,1.5) [choice] (00) {};
	\node at (0,0) [choice] (01) {};
	\node at (2,0) [choice] (20) {};
	\node at (2,1) [choice] (21) {};
	\node at (2,2) [choice] (22) {};
	\node at (00) [location1] {};
	\node at (22) [location1] {};
	\node at (01) [location2] {};
	\node at (20) [location2] {};
	\draw[thick] (00) to (20);	
	\draw[thick] (01) to (20);	
	\draw[thick] (01) to (21);
	\draw[thick] (01) to (22);
	\node at (-0.7,1.5) {\small $a$};
	\node at (-0.7,0) {\small $b$};
	\node at (2.7,2) {\small $c$};
	\node at (2.7,1) {\small $d$};
	\node at (2.7,0) {\small $e$};
	\node at (0,2.2) [color=black!77] {\footnotesize
	};
	\node at (2,2.7) [color=black!77] {\footnotesize
	};
	\node at (0,-0.7) [color=black!77] {\footnotesize
	};
	\node at (2,-0.7) [color=black!77] {\footnotesize
	};
\end{tikzpicture}
\qquad
\begin{tikzpicture}[scale=0.55,choice/.style={draw, circle, fill=black!100, inner sep=2.1pt},
	location1/.style={draw, very thick, densely dotted, circle, color=black!77, inner sep=3.2pt},
	location2/.style={draw, thick, circle, color=black!77, inner sep=3.2pt}]
	\node at (-2.5,2) {\small\bf $(G',\mathcal{H}_2')$:};	
	\node at (2,0) [choice] (00) {};
	\node at (2,1.5) [choice] (01) {};
	\node at (0,0) [choice] (20) {};
	\node at (0,1) [choice] (21) {};
	\node at (0,2) [choice] (22) {};
	\node at (00) [location1] {};
	\node at (22) [location1] {};
	\node at (01) [location2] {};
	\node at (20) [location2] {};
	\draw[thick] (00) to (20);	
	\draw[thick] (01) to (20);	
	\draw[thick] (01) to (21);
	\draw[thick] (01) to (22);
	\node at (2.7,1.5) {\small $u$};
	\node at (2.7,0) {\small $v$};
	\node at (-0.7,2) {\small $r$};
	\node at (-0.7,1) {\small $s$};
	\node at (-0.7,0) {\small $t$};
	\node at (0,2.7) [color=black!77] {\footnotesize
	};
	\node at (2,-0.7) [color=black!77] {\footnotesize
	};
	\node at (0,-0.7) [color=black!77] {\footnotesize
	};
	\node at (2,2.2) [color=black!77] {\footnotesize
	};
\end{tikzpicture}
\end{center}
\end{example}

\begin{example}\label{ex: focal points}
We consider two concrete examples of focal points. However, before that, 
note that if choice $c_i$ of player $i$ in stage $S$ is a focal point, then one of
the following two scenarios hold by the definition of focal points:
\begin{itemize}
    \item 
    $c_i$ is not structurally equivalent to any other
    choice in stage $S$.
    \item
    $c_i$ is structurally equivalent to some other 
    choices $d_1,\dots,d_{\ell}$ in $S$. In this case \emph{all 
    the choices $c_i, d_1,\dots, d_{\ell}$ must belong to
    the same single edge of the winning relation
    $W_G$} for the following reason: the choice 
    $c_i$ is structurally equivalent to the 
    choice $d_j\in\{d_1,\dots, d_{\ell}\}$ of player $j$
    but $c_i$ is \emph{not} structurally equivalent to 
    any other choice of player $i$, so $d_j$ cannot be structurally equivalent to 
    any other choice of player $j$. 
\end{itemize}

Now to the examples. Consider the first two rounds of the game $\CM_5$, pictured below,
where the players fail to coordinate by
first selecting the pair $(a_1,b_2)$ and then fail 
again by selecting the pair $(b_1,c_2)$. 

\vspace{-4mm}

\begin{center}
\begin{tikzpicture}[scale=0.6,choice/.style={draw, circle, fill=black!100, inner sep=2.1pt},
	location1/.style={draw, very thick, densely dotted, circle, color=black!77, inner sep=3.2pt},
	location2/.style={draw, thick, circle, color=black!77, inner sep=3.2pt}]
	\node at (-3,4) {$\CM_5:$};	
	\node at (0,0) [choice] (00) {};
	\node at (2,0) [choice] (20) {};
	\node at (0,1) [choice] (01) {};
	\node at (2,1) [choice] (21) {};
	\node at (0,2) [choice] (02) {};
	\node at (2,2) [choice] (22) {};
	\node at (0,3) [choice] (03) {};
	\node at (2,3) [choice] (23) {};
	\node at (0,4) [choice] (04) {};
	\node at (2,4) [choice] (24) {};
    \node at (04) [location1] {};
    \node at (23) [location1] {};
    \node at (03) [location2] {};
    \node at (22) [location2] {};
	\draw[thick] (00) to (20);
	\draw[thick] (01) to (21);
	\draw[thick] (02) to (22);
	\draw[thick] (03) to (23);
	\draw[thick] (04) to (24);
	\node at (-0.8,0) {\small $e_1$};
	\node at (-0.8,1) {\small $d_1$};
	\node at (-0.8,2) {\small $c_1$};
	\node at (-0.8,3) {\small $b_1$};	
	\node at (-0.8,4) {\small $a_1$};	
	\node at (2.8,0) {\small $e_2$};
	\node at (2.8,1) {\small $d_2$};
	\node at (2.8,2) {\small $c_2$};
	\node at (2.8,3) {\small $b_2$};	
	\node at (2.8,4) {\small $a_2$};	
\end{tikzpicture}
\end{center}
The structural equivalence classes become 
modified in this scenario as follows:
\begin{itemize}
\item Initially all choices are structurally equivalent.
\item After the first round, the equivalence classes are $\{a_1,b_2\}$, $\{b_1,a_2\}$ and $\{c_1,d_1,e_1,c_2,d_2,e_2\}$.
\item After the second round, the equivalence classes are $\{a_1\}$, $\{a_2\}$, $\{b_1\}$, $\{b_2\}$, $\{c_1\}$, $\{c_2\}$ and $\{d_1,e_1,d_2,e_2\}$.
\end{itemize}
There are no focal points in the initial stage $S_0$ and the
same is true for the next stage $S_1$. However, in the stage $S_2$, all the choices $a_1,b_1,c_1,a_2,b_2,c_2$ become focal points, and the players can thus immediately guarantee coordination in the third round by selecting any winning pair of focal points, i.e., any of the pairs $(a_1,a_2)$, $(b_1,b_2)$, $(c_1,c_2)$.
(We note that, from the point of view of the general study of
rational choice, it may not be obvious which of these
pairs should selected, so a convention may be needed to fix which protocol to use.)

Consider then the game $\overline{G(O_5)}$, the 
complement of the cycle game $G(O_5)$.
%
%
In the pictures below, we present $\overline{G(O_5)}$ also in the
form where the choices are
arranged in a cycle and we draw the choices of player 2 in white for clarity.

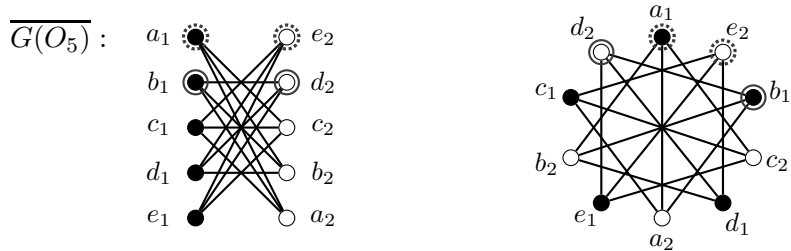
\begin{center}
\begin{tikzpicture}[scale=0.6,choice/.style={draw, circle, fill=black!100, inner sep=2.1pt},
    choice2/.style={draw, circle, inner sep=2.1pt},
	location1/.style={draw, very thick, densely dotted, circle, color=black!77, inner sep=3.2pt},
	location2/.style={draw, thick, circle, color=black!77, inner sep=3.2pt}]
	\node at (-3,4) {$\overline{G(O_5)}:$};	
	\node at (0,0) [choice] (00) {};
	\node at (2,0) [choice2] (20) {};
	\node at (0,1) [choice] (01) {};
	\node at (2,1) [choice2] (21) {};
	\node at (0,2) [choice] (02) {};
	\node at (2,2) [choice2] (22) {};
	\node at (0,3) [choice] (03) {};
	\node at (2,3) [choice2] (23) {};
	\node at (0,4) [choice] (04) {};
	\node at (2,4) [choice2] (24) {};
	\draw[thick] (00) to (22);
	\draw[thick] (00) to (23);
	\draw[thick] (00) to (24);
	\draw[thick] (01) to (21);
	\draw[thick] (01) to (23);
	\draw[thick] (01) to (24);
	\draw[thick] (02) to (20);
	\draw[thick] (02) to (22);
	\draw[thick] (02) to (24);
	\draw[thick] (03) to (20);
	\draw[thick] (03) to (21);
	\draw[thick] (03) to (23);
	\draw[thick] (04) to (20);
	\draw[thick] (04) to (21);
	\draw[thick] (04) to (22);
    \node at (04) [location1] {};
    \node at (24) [location1] {};
    \node at (03) [location2] {};
    \node at (23) [location2] {};	
	\node at (-0.8,0) {\small $e_1$};
	\node at (-0.8,1) {\small $d_1$};
	\node at (-0.8,2) {\small $c_1$};
	\node at (-0.8,3) {\small $b_1$};	
	\node at (-0.8,4) {\small $a_1$};	
	\node at (2.8,0) {\small $a_2$};
	\node at (2.8,1) {\small $b_2$};
	\node at (2.8,2) {\small $c_2$};
	\node at (2.8,3) {\small $d_2$};	
	\node at (2.8,4) {\small $e_2$};
    \node at (0,-0.7) {};
\end{tikzpicture}
\hspace{21mm}
\begin{tikzpicture}[scale=0.4,choice/.style={draw, circle, fill=black!100, inner sep=2.1pt},
    choice2/.style={draw, circle, inner sep=2.1pt},
	location1/.style={draw, very thick, densely dotted, circle, color=black!77, inner sep=3.2pt},
	location2/.style={draw, thick, circle, color=black!77, inner sep=3.2pt}]
%
	\node at (3,0) [choice2] (00) {};
	\node at (5,0.5) [choice] (20) {};
	\node at (6,2) [choice2] (01) {};
	\node at (1,0.5) [choice] (21) {};
	\node at (0,2) [choice2] (02) {};
	\node at (6,4) [choice] (22) {};
	\node at (5,5.5) [choice2] (03) {};
	\node at (0,4) [choice] (23) {};
	\node at (1,5.5) [choice2] (04) {};
	\node at (3,6) [choice] (24) {};
	\draw[thick] (00) to (22);
	\draw[thick] (00) to (23);
	\draw[thick] (00) to (24);
	\draw[thick] (01) to (21);
	\draw[thick] (01) to (23);
	\draw[thick] (01) to (24);
	\draw[thick] (02) to (20);
	\draw[thick] (02) to (22);
	\draw[thick] (02) to (24);
	\draw[thick] (03) to (20);
	\draw[thick] (03) to (21);
	\draw[thick] (03) to (23);
	\draw[thick] (04) to (20);
	\draw[thick] (04) to (21);
	\draw[thick] (04) to (22);
    \node at (24) [location1] {};
    \node at (03) [location1] {};
    \node at (04) [location2] {};
    \node at (22) [location2] {};
	\node at (5.5,-0.2) {\small $d_1$}; 
	\node at (6.8,1.8) {\small $c_2$};
	\node at (6.9,4.2) {\small $b_1$}; 
	\node at (5.5,6.2) {\small $e_2$}; 
	\node at (3,6.8) {\small $a_1$};
	\node at (3,-0.8) {\small $a_2$};
	\node at (0.5,0) {\small $e_1$};
	\node at (-0.8,1.8) {\small $b_2$};
	\node at (-0.8,4.2) {\small $c_1$};	
	\node at (0.4,6.3) {\small $d_2$};    
\end{tikzpicture}
\end{center}

Note that all choices are initially structurally equivalent in $\overline{G(O_5)}.$ Suppose then that the players fail to coordinate in the first round. This can happen only if they select choices that are ``adjacent in the cycle'' (see the picture above). Hence, by symmetry, we may assume that the players choose the pair $(a_1,e_2)$ in
the first round. Then the equivalence classes after 
the first round are $\{a_1,e_2\}$, $\{b_1,d_2\}$, $\{c_1,c_2\}$, $\{d_1,b_2\}$ and $\{e_1,a_2\}$. Hence the players can guarantee coordination in the second round by selecting the winning pair $(b_1,d_2)$ of focal points (or alternatively the pair $(c_1,c_2)$ or $(d_1,b_2)$).
\end{example}

\begin{example}\label{ex: use of LA}
We illustrate the use of the \LA protocol in the game $\CM_5$, pictured below.
Suppose that coordination fails in the first round. By symmetry, we may assume that the players selected $a_1$ and $b_2$. Now, in the resulting stage $S_1$, the structural equivalence classes are $\{a_1,b_2\}$, $\{b_1,a_2\}$ and $\{c_1,d_1,e_1,c_2,d_2,e_2\}$.

If the pair $(b_1,a_2)$ is selected in the next round, then the structural equivalence classes do not change and thus the resulting next stage is automorphism-equivalent to $S_1$. Hence, by following \LA, player 1 should avoid selecting $b_1$ and player 2 should avoid selecting $a_2$. For the same reason, the players should also avoid selecting the choices $a_1$ and $b_2$.

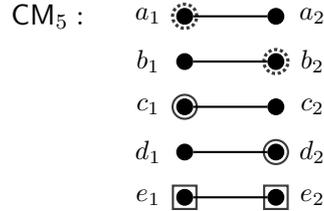
\begin{center}
\begin{tikzpicture}[scale=0.6,choice/.style={draw, circle, fill=black!100, inner sep=2.1pt},
	location1/.style={draw, very thick, densely dotted, circle, color=black!77, inner sep=3.2pt},
	location2/.style={draw, thick, circle, color=black!77, inner sep=3.2pt},
	location3/.style={draw, thick, rectangle, color=black!77, inner sep=4.4pt}]
	\node at (-3,4) {$\CM_5:$};	
	\node at (0,0) [choice] (00) {};
	\node at (2,0) [choice] (20) {};
	\node at (0,1) [choice] (01) {};
	\node at (2,1) [choice] (21) {};
	\node at (0,2) [choice] (02) {};
	\node at (2,2) [choice] (22) {};
	\node at (0,3) [choice] (03) {};
	\node at (2,3) [choice] (23) {};
	\node at (0,4) [choice] (04) {};
	\node at (2,4) [choice] (24) {};
    \node at (04) [location1] {};
    \node at (23) [location1] {};
    \node at (02) [location2] {};
    \node at (21) [location2] {};
    \node at (00) [location3] {};
    \node at (20) [location3] {};
	\draw[thick] (00) to (20);
	\draw[thick] (01) to (21);
	\draw[thick] (02) to (22);
	\draw[thick] (03) to (23);
	\draw[thick] (04) to (24);
	\node at (-0.8,0) {\small $e_1$};
	\node at (-0.8,1) {\small $d_1$};
	\node at (-0.8,2) {\small $c_1$};
	\node at (-0.8,3) {\small $b_1$};	
	\node at (-0.8,4) {\small $a_1$};	
	\node at (2.8,0) {\small $e_2$};
	\node at (2.8,1) {\small $d_2$};
	\node at (2.8,2) {\small $c_2$};
	\node at (2.8,3) {\small $b_2$};	
	\node at (2.8,4) {\small $a_2$};	
\end{tikzpicture}
\end{center}

Hence, by following \LA in $S_1$, the players will select among the set $\{c_1,d_1,e_1,c_2,d_2,e_2\}$ with the uniform probability distribution.
Supposing that they fail again in coordination, we may assume by symmetry that they selected the pair $(c_1,d_2)$.
The equivalence classes in the resulting stage $S_2$ are $\{a_1,b_2\}$, $\{b_1,a_2\}$, $\{c_1,d_2\}$, $\{d_1,c_2\}$ and $\{e_1,e_2\}$. Now, selecting any of the pairs $(a_1,b_2)$, $(b_1,a_2)$, $(c_1,d_2)$ and $(d_1,c_2)$ leads to a next stage which is automorphism-equivalent to~$S_2$. Thus, by following \LA in $S_2$, the players will select the pair $(e_1,e_2)$. This leads to guaranteed coordination in the third round.
\end{example}

\begin{definition}\label{wmseconddef}
The \defstyle{wait-or-move protocol (\WM)} for repeated two-player \WLC-games 
goes as follows. Pick your first choice randomly (with a uniform
probability over all choices). Then do the following in all non-final stages.

\begin{enumerate}
\item[(1)] Suppose that both players have selected only a single choice (possibly several times) in the previous rounds. Let $c_1$ be your earlier choice and $c_2$ the earlier choice of the other player.
Then select your next choice according to the
probability distribution $f$ such that
\begin{itemize}
    \item $f(c_1)=\frac{1}{2}$, and
    \item each choice that coordinates with $c_2$ 
     is picked with equal probability, the
     total probability over such choices being $\frac{1}{2}$.
\end{itemize}

\item[(b)] Suppose that both players have selected exactly two choices (possibly several times). Then select one of your previous choices, each with
probability $\frac{1}{2}$.
\item[(c)] In any other non-final stage, pick your choice randomly. (Note that
such a non-final stage cannot even be reached if both players follow \WM.)
\end{enumerate}
\end{definition}


\section{Appendix: Complete proofs and additional lemmas}\label{appendix: Proofs}

\noindent
\textbf{Proposition \ref{the: LG3} restated.}
\emph{$\LA$ is the uniquely \ECT-optimal and uniquely \GCT-optimal in $\CM_3$.}

\begin{proof}
Every structural protocol (and thus \LA) must choose a random choice in
the first round of $\CM_3$. If the players fail to
coordinate in the first round, then the only probability
distribution that guarantees a win in the second round
selects the unique choice from the only
untouched edge (as \LA instructs). It is thus clear that \LA is both
uniquely \ECT-optimal and uniquely \GCT-optimal in $\CM_3$.
\end{proof}

\noindent
\textbf{Proposition \ref{the: GCT with RES} restated.}
\emph{$\LA$ guarantees coordination in games $\CM_m$ in $\lceil m/2 \rceil$ rounds when $m$ is odd, 
but $\LA$ does not guarantee coordination in $\CM_m$ for any even $m$.}

\begin{proof}
For the sake of completeness, we give here a full proof of
the proposition. However, the fact that \emph{$\LA$ does not
guarantee coordination in $\CM_m$ for any even $m$}
will also follow directly from Theorem \ref{the: even m}, whose proof
does not depend in any way of the current proposition.

Consider $\CM_m$ with an odd $m$. As $\LA$ is a structural
protocol, the players must pick randomly in the first round.
Supposing they do not coordinate, this creates two touched edges.
In the next round, the players must pick choices that are not on
the touched edges, because
the protocol \LA instructs to pick---if possible---choices that cannot
lead to a stage that is automorphism-equivalent to the current stage.
Similarly, in every round where the players have failed to coordinate,
they must choose from untouched edges. In the worst case, since $m$ is odd,
the players can fail to coordinate until there is exactly one
untouched edge left. Then the players coordinate in the next stage, and
clearly this takes $\lceil m/2 \rceil$ rounds.

The scenario is very similar in the case $m$ is even, but this
time the players may end up in a situation where all edges have 
become touched, but coordination has failed. Then every choice $c$ of
player $1$ is structurally equivalent to a choice $c'$ of player $2$
such that $c$ and $c'$ are not on the same edge. Since the players are
using a structural protocol, they may end up choosing such a pair of
structurally equivalent choices, failing coordination. Moreover, if
the players indeed choose the pair $(c,c')$,
this leads an automorphism-equivalent stage.
Therefore, picking $(c,c')$ leads to the same problem again:
coordination can fail due to picking structurally equivalent choices.
This way the players may end up forever 
choosing automorphism-equivalent stages without coordinating.
\end{proof}

\noindent
\textbf{Theorem \ref{jormatheorem} restated.}
\emph{Let $G$ be a \WLC-game with
one-shot coordination probability $p$
when both players make their first choice randomly.
Then the expected coordination time by \WM is at most $3-2p$.} 
%

\begin{proof}
Let $G$ be a \WLC-game.
First note that the probability for coordination with \WM in the first round is
trivially $p$. After that, the players follow \WM and thus, in every round,
either repeat their previous choice or swap to another choice with
equal probability. If one player repeats and the other one swaps, then they coordinate.
Thus, in every round after the first round,
the one-shot coordination probability is at least $\frac{1}{2}$.
(Note that this probability can be greater than $\frac{1}{2}$ as $G$ is not 
necessarily a choice matching game.)
Let us first consider the case where $G$ is 
such that if coordination fails in the first round,
then, in every subsequent round, the probability of
coordination is exactly $\frac{1}{2}$. (This includes, e.g., all
choice matching games $\CM_m$ with $m > 1$.) Now the probability of
coordinating in the $k$th round (and not earlier) is
%
%
$(1 - p) \cdot (\frac{1}{2})^{k-1}$
%
%
for all $k\geq 2$. 
Hence the expected value $E$ for the coordination time with \WM is
calculated as follows.

$$
E\ =\ p\ +\ (1 - p)\sum\limits_{ k\geq 2}^{\infty} \dfrac{k}{2^{k-1}}
\ =\ p\ +\ (1 - p)\sum\limits_{ k\geq 2}^{\infty} \dfrac{2k}{2^{k}}\ .
%
$$
%
%


\noindent
It is well known that
$$\sum\limits_{ k\geq 1}^{ \infty} \dfrac{k}{2^{k}} =2,$$
whence
$$\sum\limits_{ k\geq 2}^{ \infty}
\dfrac{k}{2^{k}} = \dfrac{3}{2}\ 
\text{ and thus }\sum\limits_{ k\geq 2}^{\infty} {\dfrac{2k}{2^{k}} =  3}.$$

\vcut{
Now, it is well known and the reader 
can check that $\sum\limits_{k\geq 1}^{\infty} \cfrac{k}{2^{k}} = 2$.


\noindent
Thus $\sum\limits_{k\geq 2}^{\infty} \cfrac{k}{2^{k}} = \cfrac{3}{2}$ 
and therefore $\sum\limits_{k\geq 2}^{\infty} \cfrac{2k}{2^{k}} =  3$.
}

\noindent
Thus $E = p + (1 - p)\cdot 3  = 3 - 2p$. Therefore, in the general
case where the probability of coordinating is \emph{at most} $\frac{1}{2}$ in
the rounds after the first one, it is now
immediate that $E \leq 3 - 2p$ and thus $3 - 2p$ is still an
upper bound for the expected coordination time with \WM.
\vcut{
We then consider the remaining case
where some nodes $u_1,v_1$ 
such that $(u_1,v_1)\not\in W_G$ have 
their winning extensions connected by an edge.
Above the
probability for coordinating in round $k\geq 2$ on the condition of
not coordinating earlier was $1/2$. Now this conditional probability is higher, and
therefore it is immediate that $3 - 2p$ is still an
upper bound for expected coordination time.
}
\end{proof}

\noindent
\textbf{Proposition \ref{the: LG2} restated.}
\emph{\WM is uniquely \ECT-optimal in $\CM_2$.}

\medskip

\noindent
\begin{proof} The claim
follows directly from Lemma \ref{Two-edge lemma} given below;
we will first present a technical auxiliary definition
(Definition \ref{conjugatedefinition}) and
then prove Lemma \ref{Two-edge lemma}.
\end{proof}

\begin{definition}\label{conjugatedefinition}
Consider a choice 
matching game $\CM_m$ and assume a stage $(\CM_m,\mathcal{H}_k)$ has
edges $(u,v)$ and $(u',v')$ such that $u \sim v'$ and $u' \sim v$ (recall Definition~\ref{def: renamings}).
Then we say that the nodes $u$ and $u'$ are \textbf{conjugates} (of each other), and
likewise, the choices $v$ and $v'$ are conjugates.
The following lemma states that protocols become faster if they are 
adjusted to assign the same probability to conjugate elements in 
choice matching games.
\end{definition}

\begin{lemma}\label{Two-edge lemma}
Let $S=(\CM_m,\mathcal{H}_k)$ be a stage and $\pi$ be a protocol
which assigns
different probabilities $p_u$ and $p_{u'}$ to
some conjugate nodes $u$
and $u'$ of $S$ (see Definition \ref{conjugatedefinition} above). 
Let $\pi'$ be the protocol that is
otherwise as $\pi$ but assigns $u$ and $u'$ the same 
probability $\frac{1}{2}(p_u + p_{u'})$ in~$S$.
Then the \ECT from $S$ with $\pi'$ is
strictly smaller than the \ECT from $S$ with $\pi$. 
\end{lemma}

\begin{proof}
Let $v$ and $v'$ denote the choices such that $\CM_m$ has
edges $(u,v)$ and $(u',v')$.
As $\pi$ is a structural protocol, we must have $p_{v'} = p_{u}$
and $ p_{v} = p_{u'}$. To simplify notation, 
call $p_u = x$ and $p_u + p_{u'} = c$. Thus $p_v = p_{u'} = c-p_u = c - x $.

Under the condition that both players end up choosing from one of
the edges $(u,v),(u',v')$ in the stage $S$, the probability of winning is
$$\dfrac{2\cdot x(c-x)}{c^2} = -\dfrac{2x^2}{c^2} + \dfrac{2x}{c}.$$
This has its global maximum at $x = \frac{c}{2} = 
\frac{1}{2}(p_u + p_{u'})$. Since $\pi$ and $\pi'$ agree on
all moves other than the one discussed here, the claim follows.
\end{proof}


\medskip

\noindent
\textbf{Theorem \ref{the: odd m} restated.}
\emph{For all odd $m\geq 1$, \LA is uniquely \GCT-optimal in $\CM_m$.}

\begin{proof}
Let $m$ be odd. 
Recall that, by Proposition~\ref{the: GCT with RES}, the \GCT in $\CM_m$ with \LA
is $\lceil m/2\rceil$ rounds.
We assume, for contradiction, that there is
some protocol $\pi\neq\LA$ that
guarantees coordination in $\CM_m$ in at
most $\lceil m/2\rceil$ of rounds, possibly less.
As $\pi\neq\LA$, there exists some play of $\CM_m$ where both
players follow $\pi$, and in some round, at
least one of the players chooses a node on a touched edge.
(Recall from the proof of
Proposition~\ref{the: GCT with RES} that \LA never
chooses from a touched edge in \CM-game with an odd number of edges.)
Now, let $S_\ell=(\CM_m,\mathcal{H}_\ell)$
be the first stage of that play 
when this happens---so if $(c,c')$ is the most
recently recorded pair of choices in $S_\ell$, then at 
least one of $c$ and $c'$ is part of an
edge that has already been touched in some earlier round.
And furthermore, in all 
stages $S_{\ell'}$ with $\ell'< \ell$, the most 
recently chosen pair does not contain a choice belonging to an
edge that was touched in some yet earlier round $\ell''<\ell'$.

In the stage $S_{\ell -1}$ it therefore holds that
for every choice profile $(c_i,d_i)$, chosen in some round $i\leq (\ell - 1)$,
the nodes $c_i$ and $d_i$ are structurally equivalent.
Of course also the nodes of $S_{\ell - 1}$ on so far
untouched edges are structurally 
equivalent to each other. Furthermore, the number of
already touched edges in $S_{\ell - 1}$ is the
\emph{even} number $m' = 2(\ell - 1)$.

We will now show that $\pi$ does \emph{not} guarantee a 
win in $\lceil m/2 \rceil - (\ell - 1)$ rounds when starting 
from the stage $S_{\ell - 1}$. This
completes the proof, contradicting the assumption that $\pi$
guarantees a win in $\CM_m$ in at most $\lceil m/2 \rceil$ rounds.

%
Now, recall the stage $S_\ell$ from above where $(c,c')$
contained a choice from an already touched edge. By symmetry, we
may assume that $c$ is such a choice.
Starting from the stage $S_{\ell - 1}$, consider a newly
defined stage $S_\ell'$ where the 
first player again makes the choice $c$ but
the other player this time \emph{makes a structurally equivalent choice $c^*\sim c$}.
This is possible as $\pi$ is a structural protocol. Now note that
the choice profile $(c,c^*)$ is not winning since $c$ and $c^*$ are structurally 
equivalent choices from already touched edges, and thus either $(c,c^*)$
is a choice profile that has already been chosen in some
earlier round $j < \ell$, or the nodes $d,d^*$
adjacent in $\CM_m$ to $c^*,c$ (respectively) form a choice profile $(d,d^*)$
chosen in some earlier round $j < \ell$.
Therefore, in the freshly defined stage $S_\ell'$, the players have in every stage
(including the stage $S_\ell'$ itself) selected a choice profile that consists of two
structurally equivalent choices. Both choices in the most recently selected
choice profile in $S_\ell'$ have been picked from edges that have become 
touched even earlier. It now suffices to
show that it can still take $\lceil m/2 \rceil - (\ell - 1)$ rounds to
finish the game. To see that this is the case, 
we shall next consider a play from the stage $S_\ell'$
onwards where in each remaining round, the choice profile $(e,e^*)$
picked by the players consists of structurally equivalent choices; such a play
exists since $\pi$ is structural.
Due to picking only structurally 
equivalent choices in the remaining play,
when choosing a profile from the already touched 
part, the players will clearly never coordinate. 
And when choosing from the untouched part, immediate coordination is
guaranteed if and only if there is only one untouched edge left.
Therefore the players coordinate exactly when 
they ultimately select from the last untouched edge. As the stage $S_{l}'$
has precisely $m - 2(\ell - 1)$ untouched edges, winning in this 
play takes at least
$$\bigg\lceil\ \frac{m - 2(\ell - 1)}{2}\bigg\rceil\ 
=\ \lceil m/2 \rceil - (\ell - 1)$$
rounds to win from $S_\ell'$.
\end{proof}

\noindent
\textbf{Lemma \ref{32 Lemma} restated.}
\emph{The \ECT from $(\CM_m,\mathcal{H}_k)$
with no focal point is at least $\frac{3}{2}$ with any protocol.}

\begin{proof}
If $(\CM_m,\mathcal{H}_k)$ has an even number of edges, then,
since $(\CM_m,\mathcal{H}_k)$ has no
focal points, we can partition its edges into doubleton
sets, each set containing
exactly two edges $(u,v)$ and $(u',v')$ such that $u\sim v'$ and $u'\sim v$ (whence
$u$ and $u'$ as well as $v$ and $v'$ are conjugates in the 
sense of Definition \ref{conjugatedefinition}).
If $(\CM_m,\mathcal{H}_k)$ has an odd number of edges and no focal points,
then we can construct a
partition consisting of similar doubletons together with
one tripleton set with edges $(u,v),(u',v'),(u'',v'')$ such that all the choices $u,u',u'',v,v',v''$ are
all pairwise structurally equivalent.

To coordinate in the next round $k+1$, the players must 
select from the same (doubleton or tripleton) set $T$ of edges in the partition, and
within $T$, they must choose the same edge. Now recall that the players use the
same protocol, and the protocol determines the
same probability for all structurally equivalent choices. Thus 
the probability of hitting the same edge on the condition that the players 
have chosen from the same doubleton set $T$ is at most $\frac{1}{2}$ (this 
follows easily from the proof of Lemma \ref{Two-edge lemma}). 
The probability of hitting the same edge on the condition that the players choose from
the tripleton set is necessarily $\frac{1}{3}$, as all the six choices within
that tripleton are pairwise structurally equivalent, and thus the protocol assigns
them the same probabilities.
Therefore, for any protocol, $\frac{1}{2}$ is an upper bound for the probability of 
coordinating in the next round $k+1$.

Now, suppose that the players coordinate with probability $\frac{1}{2}$ in round $k+1$, and 
suppose they are guaranteed to coordinate in round $k+2$ if they fail in round $k+1$.
Then the \ECT for the remaining game is $\frac{1}{2} + \frac{1}{2}\cdot 2 = \frac{3}{2}.$
\end{proof}

The following lemma will be needed in
the proof of Theorem \ref{the: 4-ladders} below.

\begin{lemma}\label{emulationlemma}
Let $S_k = (\CM_4,\mathcal{H}_k)$ and $S_n' = (\CM_4,\mathcal{H}_n')$ be
%
%
stages of $\CM_4$ with exactly $2$ and $4$
touched edges, respectively, 
and no focal points.
Assume also, for technical convenience, that $S_n'$ does not 
extend the history of $S_k$, i.e., $S_n'$ cannot be reached from $S_k$.
Let $\pi$ be a protocol whose \ECT is $r$ 
when starting from $S_n'$. Then there exists a protocol $\pi'$ 
whose \ECT is $s \leq r$
when starting from $S_n'$ and also when starting from $S_k$.
In every stage, $\pi'$ assigns the same probability to 
conjugate nodes (cf. Definition \ref{conjugatedefinition}).
\end{lemma}

\begin{proof}
By Lemma \ref{Two-edge lemma}, it is easy to see that there 
exists a protocol $\pi^*$ whose \ECT when starting from $S_n'$ is some
number $s\leq r$ and the following conditions hold:
\begin{enumerate}
    \item
    In the stage $S_n'$ and in stages extending the history of $S_n'$, 
    the protocol $\pi^*$ always assigns the same probability to
    all nodes that are conjugates (cf. Definition \ref{conjugatedefinition}).
    \item
    Whenever a focal point is created, the protocol $\pi^*$ forces the players to 
    coordinate immediately in the next round.
\end{enumerate}
Due to the first condition above, it is possible to copy the behaviour of $\pi^*$ 
starting from $S_n'$ to all games starting
from $S_k$ in the direct way described next.
First note that both $S_n'$ and $S_k$ are based on the 
same graph $\CM_4$ with the same set of nodes.
We may assume, by symmetry, that conjugate nodes in $S_k$ are also
conjugates in $S_n'$.
We copy the behaviour of $\pi^*$ in the games 
starting from $S_n'$ to the games starting from $S_k$ just by assigning
the exact same probabilities chosen in $S_{n+\ell}'$ to the exactly 
same nodes in the corresponding
stage $S_{k+\ell}$ (that extends the history of $S_k$ in the same
way as $S_{n+\ell}'$ extends the history of $S_n'$). It is
easy to see that this constructs a structural protocol due to 
the condition 1 above stating that $\pi^*$ gives the same probabilities to
conjugate nodes. Clearly the copied protocol gives the same \ECT 
starting from $S_k$ as $\pi^*$ gives when starting from $S_n'$.

Now the ultimate 
desired protocol $\pi'$ is constructed by combining $\pi^*$ and 
the constructed copy. The assumption that $S_n'$ does not
extend the history of $S_k$ is used in this combination step.
Note that thus $\pi'$ clearly assigns the same 
probabilities to conjugates in all stages reachable from $S_k$
and $S_n'$, and by Lemma \ref{Two-edge lemma}, we can ensure that $\pi'$
also assigns the same probability to conjugates in all other stages.
\end{proof}

\noindent
\textbf{Theorem \ref{the: 4-ladders} restated.}
\emph{\WM is \ECT-optimal for $\CM_4$, but there are
continuum many other protocols that are also \ECT-optimal.}

\begin{proof}
Consider a stage $S_k=(\CM_4,\mathcal{H}_k)$ with
exactly two touched edges.
We will first show that no protocol gives an \ECT less than $2$ 
from the stage $S_k$.
This is done by establishing that existence of such a protocol
would imply existence of a
protocol in $\CM_2$ with \ECT less than $2$, contradicting
Proposition \ref{the: LG2} and Theorem \ref{jormatheorem}.

Now, suppose, for contradiction, that $\pi$
gives an \ECT less than $2$ when starting
from $S_k$. 
By Lemma \ref{Two-edge lemma}, we can assume that $\pi$ assigns
the same probability to conjugate
nodes (cf. Definition \ref{conjugatedefinition}) in $S_k$.
Therefore the formula (E) (see Section \ref{optimizingsection}) gives the \ECT 
for $\pi$ from $S_k$, given we plug in the 
right values for $p$, $E_1$, $E_2$ and $n$. We have $n=2$ and
the other values are
determined by $\pi$, with $E_1$ corresponding to the situation 
with two touched edges and $E_2$ to the situation with four touched edges.
We may assume that $E_1<2$ because all stages with exactly two touched 
edges are automorphism-equivalent and the \ECT
from $S_k$ (which has exactly two
touched edges) is less than $2$.
Using Lemma \ref{emulationlemma}, we see that there
exists a protocol $\pi'$ with $E_2 = E_1 < 2$ that also gives an \ECT 
less or equal to the \ECT of $\pi$ from $S_k$. And furthermore,  
the formula (E) with these fixed values $E_2 = E_1 < 2$
(and with $n=2$) gives the right value for the \ECT of $\pi'$ from $S_k$.
It is easy to prove that with these values, 
the formula (E) has its minimum
values at $p = 0$ and $p = 1$ when $p\in [0,1]$; 
for an illustration, see the graph of (E) in Figure~\ref{fig: E=2} 
for $E_1 = E_2 = 2 - \epsilon$ for some (small) $\epsilon > 0$.

\begin{figure}[h]
\begin{center}
\begin{tikzpicture}[domain=0:1,xscale=5,yscale=10,scale=0.7]
 	\draw[ultra thin,color=gray] (0,1.75) grid[xstep=0.1,ystep=0.05] (1,1.95);
	\draw[->] (0,1.75) -- (1.05,1.75) node[right] {\scriptsize $p$};
	\draw[->] (0,1.75) -- (0,1.98) node[above] {\scriptsize $E(p)$};
	\node at (0,1.72) {\tiny $0$};
	\node at (0.5,1.72) {\tiny $0.5$};
	\node at (1,1.72) {\tiny $1$};
	\node at (-0.09,1.8) {\tiny $2\!-\!\epsilon$};
	\node at (-0.09,1.9) {\tiny $2\!-\!\frac{\epsilon}{2}$};
	\draw[thick] plot (\x, {\x^2*(1/2 + (1/2)*(1 + 1.6)) + 4*\x*(1-\x) + (1-\x)^2*(1/2+(1/2)*(1+1.6))});
\end{tikzpicture}
\;
\begin{tikzpicture}[domain=0:1,xscale=5,yscale=10,scale=0.7]
 	\draw[ultra thin,color=gray] (0,1.9) grid[xstep=0.1,ystep=0.05] (1,2.1);
	\draw[->] (0,1.9) -- (1.05,1.9) node[right] {\scriptsize $p$};
	\draw[->] (0,1.9) -- (0,2.13) node[above] {\scriptsize $E(p)$};
	\node at (0,1.87) {\tiny $0$};
	\node at (0.5,1.87) {\tiny $0.5$};
	\node at (1,1.87) {\tiny $1$};
	\node at (-0.07,1.9) {\tiny $1.9$};
	\node at (-0.07,2) {\tiny $2$};
	\node at (-0.07,2.1) {\tiny $2.1$};
	\draw[thick] plot (\x, {\x^2*(1/2 + (1/2)*(1+2)) + 4*\x*(1-\x) + (1-\x)^2*(1/2+(1/2)*(1+2))});
\end{tikzpicture}
\caption{Left: Curve of (E) when $n=2$ and $E_1=E_2=2-\epsilon$ for some $\epsilon>0$. Right: Curve of (E) when $n=2$ and $E_1=E_2=2$.}
\label{fig: E=2}
\end{center}
\end{figure}
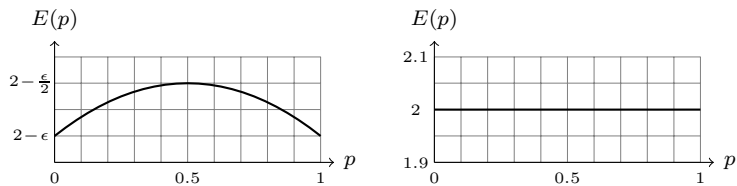

Therefore the protocol $\pi''$
that uses $p = 1$ at $S_k$, but
otherwise behaves as $\pi'$, has the following properties:
\begin{itemize}
    \item 
    $\pi''$ has the same \ECT (less than $2$) as $\pi'$ when starting from $S_k$.
    \item
    $\pi''$ shares the values $E_2 = E_2 < 2$ with $\pi'$.
\end{itemize} 
(We note that of course possibly $\pi'' = \pi'$.)
Now, $\pi''$ instructs the
players to choose from an already touched edge at $S_k$, so every
resulting stage $S_{k+1}$ turns out to be automorphism-equivalent to
the stage $S_k$.
Thus we can repeat the
reasoning above concerning $S_k$, this time beginning from $S_{k+1}$. Iterating
the argument repeatedly, it is easy to see that in the limit, we get a 
protocol that behaves precisely as \WM but has an \ECT less than $2$
when starting from $S_k$. This contradicts the 
fact that the \ECT of \WM is $2$
when starting from $S_k$ by Lemma \ref{roskalemma}.
Having proved that no protocol has an \ECT less than $2$ in $S_k$, we
then observe by Lemma \ref{roskalemma} that therefore $\WM$ is an \ECT-optimal
protocol for $S_k$ and therefore trivially also for $\CM_4$. We still 
must find continuum many other optimal protocols for $\CM_4$.
Clearly it suffices to prove that there are continuum many other optimal
protocols when starting from an arbitrary stage $S_k$ where we
have exactly two touched edges.
%

%
Consider again the formula (E) with $n=2$ and $E_1 = E_2 = 2$, i.e.,
the values given by \WM which we above
identified to be \ECT-optimal in $\CM_4$ and also 
when starting from $S_k$. It is easy to show that with 
these values, the formula (E) becomes equal to the
constant $2$ for all $p\in [0,1]$; see Figure~\ref{fig: E=2} for an 
illustration of the corresponding flat curve and its contrast to 
the case where $E_1 = E_2 = 2 - \epsilon$. Therefore we can clearly modify \WM to
give any value of $p\in[0,1]$ when starting from $S_k$ such that,
despite the modification, the resulting 
protocol is still \ECT-optimal in $\CM_4$. Thus there exist at least
continuum many \ECT-optimal protocols for $\CM_4$. In fact, it is
clear that we can analogously modify these protocols also in other stages in 
addition to $S_k$ without changing the \ECT. However, it is
straightforward to establish that the number of all protocols for $\CM_4$,
whether optimal or not, is
limited by the continuum, so there indeed exist precisely continuum 
many \ECT-optimal protocols for $\CM_4$.
\end{proof}

\section{Appendix: Analysis of ECTs in 3- and 5-choice games}\label{appendix: 3 & 5}

In this section we will systematically analyse all 3-choice games and 5-choice games and give estimates for {\ECT}s in them (recall Section~\ref{sec: characterisations} for the exact definition of an $m$-choice game).
This analysis is necessary for the special cases $m=3$ and $m=5$ in the proof of Theorem~\ref{the: upper bounds for ECTs}.
We will be using the notations for \WLC-games from Example~\ref{drawingexamples}.

We first note that the optimal \ECT is $1$ for all those $\WLC$-games in which coordination can be guaranteed in a single round. Such games are given a complete characterization in \cite{lori2017}. For example, in the game $G(1\times 1 + 2\times 2)$, coordination can be guaranteed in a single round by both players selecting choices of degree 1 (which are indeed focal points), or alternatively, by both
selecting choices of degree 2 (which form a ``winning focal set'').
%
%

\subsection{Analysis of 3-choice games}\label{sec: 3-choice games}

In this section we will show that, among all two-player $3$-choice games, the greatest optimal \ECT is uniquely realized by the game $G(1\times 2+2\times 1)$.
We also show that the optimal \ECT for this game is 
$$\frac{1+\sqrt{4+\sqrt{17}}}{2} \quad (\approx 1,925).$$

We first note that if either of the players has a choice of degree $3$ in a $3$-choice game $G$, then the optimal \ECT in $G$ is $1$ (since selecting such a choice trivially guarantees coordination).
Thus we can restrict our analysis to those $3$-choice games in which the degree of each choice is at most $2$. Note that the game graph of $G$ must thus consist of components which are either \emph{cycles} or \emph{paths} (in particular, they are subgraphs of the form $G(O_n)$, $G(1\times 1)$, $G(1\times 2)$, $G(Z_n)$, $G(\Sigma_n)$; recall the notations from Example~\ref{drawingexamples}). We list here systematically all such $3$-choice games $G$ grouped by the number of edges in the winning relation $W_G$. (Note that we must have $3\leq|W_G|\leq 6$ as $G$ is a $3$-choice game and the degree of each choice is at least $1$ and at most $2$.)

\begin{center}
{\footnotesize
\begin{tabular}{|c|c|c|c|}  
\hline
$|W_G|=3$ & $|W_G|=4$ & $|W_G|=5$ & $|W_G|=6$ \\
\hline
$G(1\times 2 + 1\times 1)$  & $G(\Sigma_3)$ & $G(O_2+1\times 1)$ & $G(O_3)$ \\
$G(3(1\times 1))=\CM_3$ & $G(Z_2+1\times 1)$ & $G(Z_3)$ & \\
& $G(1\times 2+2\times 1)$ & & \\
\hline
\end{tabular}
}
\end{center}
Among these games, the only ones that do not have a focal point are the
games $\CM_3$, $G(O_3)$ $G(1\times 2 + 2\times 1)$ which we analyse below.


\begin{itemize}
\item $\CM_3$ $(=\overline{G(O_3)})$

The optimal \ECT here is $1 + \frac{2}{3}$ by Proposition~\ref{the: LG3}
(see the table in Section \ref{sect:solvability}).

\item $G(O_3)$ $(=\overline{\CM_3})$

The one-shot coordination probability here is $\frac{2}{3}$. Suppose that the players simply make a random choice in every round (with uniform probability distribution).
The obtaind \ECT can then be calculated as follows:
\begin{align*}
\sum_{k\geq 0}^\infty \frac{2}{3}\left(\frac{1}{3}\right)^k \!(k+1)
\;=\;2\cdot\sum_{k\geq 0}^\infty \frac{k+1}{3^{k+1}}
\;=\;2\cdot\sum_{k\geq 1}^\infty \frac{k}{3^k}
\;\overset{(\star)}{=}\; 2\cdot\frac{3}{4}
\;=\;1+\frac{1}{2}.
\end{align*}
$(\star)$ It is easy to shown that $\sum_{k\geq 1}^\infty \frac{k}{3^k} = \frac{3}{4}$.

(It is relatively easy to see that this \ECT will indeed be optimal for $\overline{\CM_3}$, but there is no need for us to prove it here.) 
\item $G(1\times 2 + 2\times 1)$

We will show that the optimal \ECT for this game is $\frac{1+\sqrt{4+\sqrt{17}}}{2}$, but there are several protocols which give this optimal \ECT. See below for a proof.
\end{itemize}

\noindent
Consider the following game:

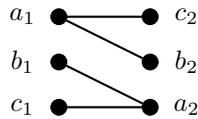
\begin{center}
\begin{tikzpicture}[scale=0.6,choice/.style={draw, circle, fill=black!100, inner sep=2.1pt}]
%
	\node at (0,0) [choice] (00) {};
	\node at (2,0) [choice] (20) {};
	\node at (0,1) [choice] (01) {};
	\node at (2,1) [choice] (21) {};
	\node at (0,2) [choice] (02) {};
	\node at (2,2) [choice] (22) {};
	\draw[thick] (00) to (20);
	\draw[thick] (01) to (20);
	\draw[thick] (02) to (21);
	\draw[thick] (02) to (22);
	\node at (-0.8,0) {\small $c_1$};
	\node at (-0.8,1) {\small $b_1$};
	\node at (-0.8,2) {\small $a_1$};
	\node at (2.8,0) {\small $a_2$};
	\node at (2.8,1) {\small $b_2$};
	\node at (2.8,2) {\small $c_2$};
\end{tikzpicture}
\end{center}

\noindent
Recalling the notion of structural equivalence from Definition \ref{def: renamings},
in the initial stage there are two structural equivalence classes: 
\begin{enumerate}
\item[(1)] $\{a_1,a_2\}$ and $\{b_1,c_1,b_2,c_2\}$.
\end{enumerate}
If players fail to coordinate by both selecting a node with degree 2, then the next stage will also be of type (1). However, if they fail to coordinate by selecting choices with degree 1, the equivalence class $\{b_1,c_1,b_2,c_2\}$ is split into two classes with two choices. We may assume by symmetry that the players chose $b_1$ and $b_2$, whence we have the following equivalence classes in the next stage:
\begin{enumerate}
\item[(2)] $\{a_1,a_2\}$, $\{b_1,b_2\}$, $\{c_1,c_2\}$.
\end{enumerate}
If players fail to coordinate by selecting the pair $(a_1,a_2)$, $(b_1,b_2)$ or $(c_1,c_2)$, then the next stage will also be of type (2). But if they fail to coordinate by one of them selecting from $\{b_1,b_2\}$ and the other one
selecting from $\{c_1,c_2\}$, then all
symmetries are broken and every
choice turns into a focal point---and thus
coordination can be guaranteed in the next round.

We first examine a stage $S_2$ of the type (2) and find the optimal probability distribution for it. The corresponding optimal \ECT will be used later for finding the optimal \ECT for a stage of type (1). 

We first observe that in order to maximize the possibility of breaking symmetries and creating focal points, it is optimal for the players to have the uniform probability distribution for selecting between the sets $\{b_1,b_2\}$ and $\{c_1,c_2\}$ (this can be proven similarly as Lemma~\ref{Two-edge lemma}).
Thus, let $p_2$ denote the probability for selecting within $\{b_1,b_2,c_1,c_2\}$. Let $E_2$ denote \ECT for the remaining game if players fail to coordinate and fail to create a focal point in $S_2$. (There are several ways how this can happen, but since all of the resulting stages are of type (2), we may assume the same \ECT for all of them by Lemma~\ref{similarity lemma}).

Under the assumptions above, the \ECT from $S_2$, with parameters $p_2$ and $E_2$, is given by the following function:
\begin{align*}
	g(p_2,E_2) 
	&= (1-p_2)^2(1+E_2) + 2p_2(1-p_2) + p_2^2\Bigl(\frac{1}{2}(1+E_2)+\frac{1}{2}\cdot 2\Bigr).\\[-1mm]
	&= \frac{1}{2}(1+3E_2)p_2^2 - 2E_2p_2 + (1 + E_2).
\end{align*}

The partial derivate $g_{p_2}=(1+3E_2) - 2E_2$ goes the zero when $p_2$ has the value $p_2^*:=\frac{2E_2}{1+3E_2}$. 
Whenever $E_2\geq 1$, the smallest value for $g(p_2,E_2)$ is obtained when $p_2=p_2^*$. 
Because both $g(p,E_2)$ and $E_2$ refer to \ECT from a stage of type (2), $E_2$ obtains its smallest possible value when
\[
	E_2 = g(p_2^*,E_2).
\]
The only (positive) solution for this equation is $E_2=\frac{3+\sqrt{17}}{4}$ $(\approx 1,781)$. This is the optimal \ECT from any stage of type (2).

\medskip

Next we will use the value $E_2$ to determine the optimal \ECT from a stage $S_1$ of type (1). Let $E_1$ denote the \ECT for the remaining game if both players select within the set $\{a_1,a_2\}$. When $p_1$ denotes the probability of choosing within the set $\{b_1,c_1,b_2,c_2\}$, the \ECT from $S_1$ is given by the following function:
\begin{align*}
	f(p_1,E_1,E_2) &= (1-p_1)^2(1+E_1) + 2p_1(1-p_1) + p_1^2(1+E_2) \\
	&= (E_1 + E_2)p_1^2 - 2E_1p_1 + (1 + E_1).
\end{align*}

The partial derivate $f_{p_1}=(2E_1 + 2E_2)p_1 - 2E_1$ goes the zero when $p_1$ has the value $p_1^*:=\frac{E_1}{E_1+E_2}$. 
Whenever $E_1,E_2\geq 1$, the
smallest value for $f(p_1,E_1,E_2)$ is obtained when $p_1=p_1^*$. 
Because both $f(p,E_1,E_2)$
and $E_1$ refer to \ECT from a stage of type (1), $E_1$
obtains its smallest possible value when $E_2=\frac{3+\sqrt{17}}{4}$ and we have
\[
	E_1 = f(p_2^*,E_2).
\]
When $E_2=\frac{3+\sqrt{17}}{4}$, the only (positive) solution for the equation above is $E_1=\frac{1+\sqrt{4+\sqrt{17}}}{2}$. This is the optimal \ECT from any stage of type (1), and thus, in particular, it is the optimal \ECT for the game $G(1\times 2 + 2\times 1)$.
Hence the greatest optimal \ECT among $3$-choice games is uniquely realized by $G(1\times 2 + 2\times 1)$. This concludes the analysis of 3-choice games.

\medskip

We digress from the main story to make a few interesting remarks.
The optimal \ECT for $G(1\times 2 + 2\times 1)$ is given by
protocols that use the optimal values for $E_1$ and $E_2$ (given above) for
calculating the probabilities $p_1^*$ $(\approx 0,5195)$
and $p_2^*$ $(\approx 0,5616)$ and use these
probabilities for selecting within the set $\{b_1,c_1,b_2,c_2\}$ 
in stages of type (1) and (2), respectively. 
However, there is no unique protocol which gives the optimal \ECT since there are 3 winning pairs of focal points that are formed if players break the symmetry in a stage of type (2).

Also note that, in $G(1\times 2 + 2\times 1)$, the optimal one-shot coordination probability (\OSCP) is $\frac{1}{2}$ and it is obtained by giving the probability $\frac{1}{2}$ for selecting a choice within the set $\{b_1,c_1,b_2,c_2\}$
(proof for this claim is similar to the proof of Lemma~\ref{Two-edge lemma}).
Since \ECT-optimal protocols for $G(1\times 2 + 2\times 1)$ do not
give the optimal \OSCP, we observe that the ``greedy protocol'' of always optimizing the chances of winning in the next round is not always \ECT-optimal. 
Another example of this phenomenon is the game $\CM_5$ where \LA does not give the optimal \OSCP in the second round; \WM is there the greedy protocol.


\subsection{Analysis of 5-choice games}\label{sec: 5-choice games}

In this section we will show that, among all two-player $5$-choice games, the greatest optimal \ECT is uniquely realized by the choice matching game $\CM_5$. Recall that this \ECT is obtained by the protocol \LA by Proposition~\ref{the: 5-ladders} and its value is $2 + \frac{1}{3}$.


We first analyse $5$-choice games $G$ for which we have $|W_G|>8$. For such games, the one-shot coordination probability $p$, when players make a random choice in the first round, is 
\[
    p \;=\; \frac{|W_G|}{|C_1||C_2|} \;\geq\; \frac{9}{25}.
\]
Thus, by Theorem~\ref{jormatheorem}, the \ECT for $G$ by following \WM is at most
\[
    3 - 2p \;\leq\; 3 - 2\cdot\frac{9}{25} 
    \;=\; 2 + \frac{7}{25} \;<\; 2 + \frac{1}{3}.
\]
Thus $G$ can be given a smaller \ECT than the optimal \ECT for $\CM_5$.

Hence we can restrict our analysis to those $5$-choice games $G$ whose winning relation $W_G$ has at most $8$ edges.
Moreover, we may also assume that neither of the player has a choice of degree $5$ as otherwise the optimal \ECT is trivially 1.

Suppose first that at least one of the players has a choice of degree $4$. Since $|W_G|\leq 8$, neither of the players can have more than two such choices and it is
impossible that both players have two such choices. If precisely one of the players has precisely one choice of degree $4$ (and the other player zero or
two such choices), then it is a focal point and the players can immediately coordinate. If one player has two choices, denoted by $c$ and $c'$, of degree $4$ and the other player has no such choice, then there are (at least three) choices that are connected to both $c$ and $c'$. Now the players can coordinate immediately by one of them selecting among $\{c,c'\}$ and the other one selecting among the choices which are connected to both $c$ and $c'$.
Finally, suppose that both players have exactly one choice of degree~$4$; we denote these by $c_1$ and $c_2$. If there is an edge between $c_1$ and $c_2$, then both of them are focal points. If there is no edge between $c_1$ and $c_2$, then we must have $G=G(1\times 4 + 4\times 1)$ as $|W_G|\leq 8$. The \ECT for this game is analysed later on below.

Suppose then that at least one of the players has a choice of degree $3$ and none of the choices have a greater degree. As $|W_G|\leq 8$, both players have at most two choices of degree~$3$. We first show that it is impossible that both players have two choices of degree $3$. If player 1 has two choices of degree $3$, then (s)he can have at most $4$ choices in total as the degree of every choice must be at least one. If also player 2 has two choices of degree $3$, then (s)he also has at most $4$ choices and thus $G$ cannot be a $5$-choice game.

We observe next that there is a focal point in $G$ if precisely one of the players has precisely one choice of degree $3$ and the other player
zero or two choices of degree $3$.
Suppose next that one player has two choices, $c$ and $c'$, of degree~$3$ and the other one has no such choices. Now there must be at least one choice which coordinates with both of $c$ and $c'$, and the players can guarantee coordination when one selects among $\{c,c'\}$ and the other one selects a choice which is connected to both $c$ and $c'$. 
Finally, suppose that both players have exactly one choice of degree $3$; these choices are denoted by $c_1$ and $c_2$. If there is an edge between $c_1$ and $c_2$, then they are focal points. If there is no edge between $c_1$ and $c_2$, then $G$ must be one of the following $5$-choice games where $|W_G|\leq 8$:

\smallskip

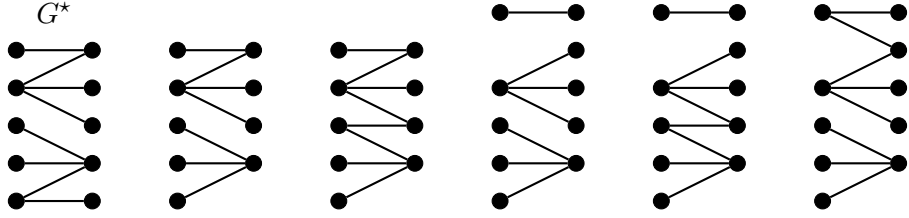
\begin{center}
\begin{tikzpicture}[scale=0.5,choice/.style={draw, circle, fill=black!100, inner sep=2.1pt}]
    \node at (1,5) {$G^\star$};
	\node at (0,0) [choice] (00) {};
	\node at (2,0) [choice] (20) {};
	\node at (0,1) [choice] (01) {};
	\node at (2,1) [choice] (21) {};
	\node at (0,2) [choice] (02) {};
	\node at (2,2) [choice] (22) {};
	\node at (0,3) [choice] (03) {};
	\node at (2,3) [choice] (23) {};
	\node at (0,4) [choice] (04) {};
	\node at (2,4) [choice] (24) {};
	\draw[thick] (00) to (21);
	\draw[thick] (01) to (21);
	\draw[thick] (02) to (21);
	\draw[thick] (03) to (22);
	\draw[thick] (03) to (23);
	\draw[thick] (03) to (24);
	\draw[thick] (00) to (20);
	\draw[thick] (04) to (24);
\end{tikzpicture}
\qquad
\begin{tikzpicture}[scale=0.5,choice/.style={draw, circle, fill=black!100, inner sep=2.1pt}]
	\node at (0,0) [choice] (00) {};
	\node at (0,1) [choice] (01) {};
	\node at (2,1) [choice] (21) {};
	\node at (0,2) [choice] (02) {};
	\node at (2,2) [choice] (22) {};
	\node at (0,3) [choice] (03) {};
	\node at (2,3) [choice] (23) {};
	\node at (0,4) [choice] (04) {};
	\node at (2,4) [choice] (24) {};
	\draw[thick] (00) to (21);
	\draw[thick] (01) to (21);
	\draw[thick] (02) to (21);
	\draw[thick] (03) to (22);
	\draw[thick] (03) to (23);
	\draw[thick] (03) to (24);
	\draw[thick] (04) to (24);
\end{tikzpicture}
\qquad
\begin{tikzpicture}[scale=0.5,choice/.style={draw, circle, fill=black!100, inner sep=2.1pt}]
	\node at (0,0) [choice] (00) {};
	\node at (0,1) [choice] (01) {};
	\node at (2,1) [choice] (21) {};
	\node at (0,2) [choice] (02) {};
	\node at (2,2) [choice] (22) {};
	\node at (0,3) [choice] (03) {};
	\node at (2,3) [choice] (23) {};
	\node at (0,4) [choice] (04) {};
	\node at (2,4) [choice] (24) {};
	\draw[thick] (00) to (21);
	\draw[thick] (01) to (21);
	\draw[thick] (02) to (21);
	\draw[thick] (03) to (22);
	\draw[thick] (03) to (23);
	\draw[thick] (03) to (24);
	\draw[thick] (02) to (22);
	\draw[thick] (04) to (24);
\end{tikzpicture}
\qquad
\begin{tikzpicture}[scale=0.5,choice/.style={draw, circle, fill=black!100, inner sep=2.1pt}]
	\node at (0,0) [choice] (00) {};
	\node at (0,1) [choice] (01) {};
	\node at (2,1) [choice] (21) {};
	\node at (0,2) [choice] (02) {};
	\node at (2,2) [choice] (22) {};
	\node at (0,3) [choice] (03) {};
	\node at (2,3) [choice] (23) {};
	\node at (2,4) [choice] (24) {};
    \node at (0,5) [choice] (05) {};
    \node at (2,5) [choice] (25) {};
	\draw[thick] (00) to (21);
	\draw[thick] (01) to (21);
	\draw[thick] (02) to (21);
	\draw[thick] (03) to (22);
	\draw[thick] (03) to (23);
	\draw[thick] (03) to (24);
    \draw[thick] (05) to (25);
\end{tikzpicture}
\qquad
\begin{tikzpicture}[scale=0.5,choice/.style={draw, circle, fill=black!100, inner sep=2.1pt}]
	\node at (0,0) [choice] (00) {};
	\node at (0,1) [choice] (01) {};
	\node at (2,1) [choice] (21) {};
	\node at (0,2) [choice] (02) {};
	\node at (2,2) [choice] (22) {};
	\node at (0,3) [choice] (03) {};
	\node at (2,3) [choice] (23) {};
	\node at (2,4) [choice] (24) {};
    \node at (0,5) [choice] (05) {};
    \node at (2,5) [choice] (25) {};
	\draw[thick] (00) to (21);
	\draw[thick] (01) to (21);
	\draw[thick] (02) to (21);
	\draw[thick] (03) to (22);
	\draw[thick] (03) to (23);
	\draw[thick] (03) to (24);
	\draw[thick] (02) to (22);
    \draw[thick] (05) to (25);
\end{tikzpicture}
\qquad
\begin{tikzpicture}[scale=0.5,choice/.style={draw, circle, fill=black!100, inner sep=2.1pt}]
	\node at (0,0) [choice] (00) {};
	\node at (0,1) [choice] (01) {};
	\node at (2,1) [choice] (21) {};
	\node at (0,2) [choice] (02) {};
	\node at (2,2) [choice] (22) {};
	\node at (0,3) [choice] (03) {};
	\node at (2,3) [choice] (23) {};
	\node at (2,4) [choice] (24) {};
    \node at (0,5) [choice] (05) {};
    \node at (2,5) [choice] (25) {};
	\draw[thick] (00) to (21);
	\draw[thick] (01) to (21);
	\draw[thick] (02) to (21);
	\draw[thick] (03) to (22);
	\draw[thick] (03) to (23);
	\draw[thick] (03) to (24);
	\draw[thick] (05) to (24);
    \draw[thick] (05) to (25);
\end{tikzpicture}
\end{center}

\smallskip

\noindent
(Note that these games have been obtained by
adding $1$ or $2$ edges and $1$ or $2$ nodes to the $4$-choice game $G(1\times 3 + 3\times 1)$.)
All the other games above, except for the leftmost game $G^\star$, have a focal point. The game $G^\star$ is analysed later on below.

We still need to analyse the case where all of the choices in $G$ have a degree at most~$2$. 
The game graph of $G$ must then consist of components which are either cycles or paths 
(cf. the corresponding case in Section~\ref{sec: 3-choice games}). 
%
%
We list here systematically all such $5$-choice games $G$ with $|W_G|\leq 8$.

\medskip

\begin{center}
{\footnotesize
\begin{tabular}{|c|c|c|c|}  
\hline
$|W_G|=5$ & $|W_G|=6$ & $|W_G|=7$ & $|W_G|=8$ \\
\hline
$G(2(1\times 2)+1\times 1)$ & $G(\Sigma_3+1\times2)$ & $G(O_2+1\times 2+1\times 1)$ & $G(O_3+1\times 2)$ \\
$G(1\times 2+3(1\times 1))$ & $G(\Sigma_3 + 2(1\times 1))$ & $G(O_2+3(1\times 1))$ & $G(O_3+2(1\times 1))$ \\
$G(5(1\times 1))=\CM_5$ & $G(Z_2+1\times 2+1\times 1)$ & $G(\Sigma_4+1\times 1)$ & $G(O_2+\Sigma_3)$ \\
& $G(Z_2+3(1\times 1))$ & $G(Z_3+1\times 2)$ & $G(O_2+Z_2+1\times 1)$ \\
& $G(2(1\times2)+2\times 1)$ & $G(Z_3+2(1\times 1))$ & $G(O_2+1\times 2+2\times 1)$  \\
& $G(1\times 2+2\times 1+2(1\times 1))$ & $G(\Sigma_3+Z_2)$ & $G(\Sigma_5)$  \\
& & $G(\Sigma_3+2\times 1+1\times 1)$ & $G(Z_4+1\times 1)$ \\
& & $G(2Z_2+1\times 1)$ & $G(\Sigma_4+2\times 1)$ \\
& & $G(Z_2+1\times 2+2\times 1)$ & $G(Z_3+Z_2)$ \\
& & & $G(\Sigma_3+\reflectbox{$\Sigma$}_3)$ \\
\hline
\end{tabular}
}
\end{center}

\medskip

\noindent
All of the the games listed above have a focal point---except for the following four games: $\CM_5$, $G(1\times 2 + 2\times 1 + 2(1\times 1))$, $G(O_3+2(1\times 1))$ and $G(\Sigma_3+\reflectbox{$\Sigma$}_3)$.

\medskip

Next we analyse the {\ECT}s for the above-identified 5-choice games $G$ whose optimal \ECT is greater than 1 and for which $|W_G|\leq 8$.

\begin{itemize}
\item $\CM_5$

The optimal \ECT here is $2 + \frac{1}{3}$ by Proposition~\ref{the: 5-ladders}
(see the table in Section \ref{sect:solvability}).

\item $G(1\times 4 + 4\times 1)$

We obtain the \ECT of $2$ rounds with the following protocol: (1) in the first round, select the choice of degree 4 with probability $\frac{1}{2}$ and some of the choices of degree 1 with the total probability $\frac{1}{2}$; (2) if coordination does not succeed, then continue with \WM. It is clear that this gives the same \ECT as \WM 
gives in the choice matching game $\CM_2$, this \ECT being $2$.
%
 
\item $G^\star$ (see the game graph given above)

As above, we obtain the \ECT of 2 rounds by first assigning the probability $\frac{1}{2}$ for selecting the choice with degree 3 and the probability $\frac{1}{2}$ for selecting the choice with degree 2, and by continuing with \WM thereafter. Again it is 
clear that this gives the same \ECT of $2$ rounds as \WM in $\CM_2$.
%

\item $G(\Sigma_3+\reflectbox{$\Sigma$}_3)$

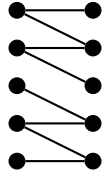
\vspace{-5mm}

\begin{center}
\begin{tikzpicture}[scale=0.5,choice/.style={draw, circle, fill=black!100, inner sep=2.1pt}]
	\node at (0,0) [choice] (00) {};
	\node at (2,0) [choice] (20) {};
	\node at (0,1) [choice] (01) {};
	\node at (2,1) [choice] (21) {};
	\node at (0,2) [choice] (02) {};
	\node at (2,2) [choice] (22) {};
	\node at (0,3) [choice] (03) {};
	\node at (2,3) [choice] (23) {};
	\node at (0,4) [choice] (04) {};
	\node at (2,4) [choice] (24) {};
	\draw[thick] (01) to (20);
	\draw[thick] (01) to (21);
	\draw[thick] (02) to (21);
	\draw[thick] (03) to (22);
	\draw[thick] (03) to (23);
	\draw[thick] (04) to (23);
	\draw[thick] (00) to (20);
	\draw[thick] (04) to (24);
\end{tikzpicture}
\end{center}
Again---for practically the same reasons as above---we obtain the \ECT 2 by first assigning the probability $\frac{1}{2}$ for selecting the choice which is ``in the middle of a 5-choice path'' and the total probability $\frac{1}{2}$ for selecting any other choice with degree 2, and by continuing with \WM thereafter.
%

\item $G(1\times 2 + 2\times 1 + 2(1\times 1))$

The players can follow an optimal protocol for $G(1\times 2 + 2\times 1)$ in the corresponding subgame and thus obtain the \ECT of less than 2 rounds (see Section~\ref{sec: 3-choice games}).
%

\item $G(O_3+2(1\times 1))$
%
%

The players can keep selecting choices randomly within the subgame $G(O_3)=\overline{\CM_3}$ to obtain the \ECT of $1+\frac{1}{2}$ rounds---as shown in Section~\ref{sec: 3-choice games}.
%
\end{itemize}

Hence we conclude that the greatest optimal expected coordination time, among all $5$-choice games, is uniquely realized by the choice matching game $\CM_5$.

\section{Appendix: Further remarks on choice matching games}\label{appendix: remarks}

In the table below we summarize the results on optimal expected and guaranteed coordination times in choice matching games $\CM_m$. The lines (---) mean that no unique protocol exists.

\begin{center}
\scalebox{0.8}{
{\footnotesize
\begin{tabular}{|c||c|c||c|c|}
 \hline
& Optimal \textbf{expected} & Unique optimal & Optimal \textbf{guaranteed} & Unique optimal \\
$m$ & coordination & protocol for & coordination & protocol for \\
& time in $\CM_m$ & expected time & time in $\CM_m$ & guaranteed time \\
\hline \rule{0pt}{1\normalbaselineskip} 
$1$\; & 1 & (any) & 1 & (any) \\[1mm]
$2$ & $2$ & WM & $\infty$ & --- \\[1mm]
$3$ & $1+\frac{2}{3}$ & \LA & $2$ & \LA \\[1mm]
$4$ & $2+\frac{1}{2}$ & --- & $\infty$ & --- \\[1mm]
$5$ & $2+\frac{1}{3}$ & \LA & $3$ & \LA \\[1mm]
\hdashline \rule{0pt}{1\normalbaselineskip} 
$6$\; & $2+\frac{2}{3}$ & WM & $\infty$ & --- \\[1mm]
$7$ & $2+\frac{5}{7}$ & WM & $4$ & \LA \\
\vdots & \vdots & \vdots & \vdots & \vdots \\
$2k$ & $3 - \frac{1}{k}$ & WM & $\infty$ & --- \\[1mm]
$2k+1$ & $3 - \frac{2}{2k+1}$ & WM & $k$ & \LA \\[1mm]
\hline
\end{tabular}
}
}
\end{center}

\medskip

%
First note that---interestingly---the game $\CM_3$ can be considered much easier than the game $\CM_2$ since the optimal \ECT is much smaller.
Moreover, coordination in $\CM_3$ can be guaranteed in two rounds, while it cannot ever be guaranteed in $\CM_2$. For similar reasons, $\CM_5$ can also be
considered easier than $\CM_4$.

In several cases there is a single unique protocol which is optimal in all aspects that we have studied in this article. In such cases one can argue that such a protocol should be followed all rational players even if they cannot communicate in advance or share any conventions.\footnote{This relies on the \emph{assumption} that the list of
\emph{possible preferences} consists of either
minimizing $\ECT$s or minimizing $\GCT$s. At least the average
case and worst case are by far the most common scenarios considered.} In the
cases where no single protocol is
optimal in all aspects, it is more
problematic for the players to choose their protocol---unless they
share some
convention.

The most clear cases here are the games with 3 and 5 (and trivially 1) choices, where the protocol \LA is uniquely optimal with respect to both \ECT and \GCT. Also all the
games with an even number of choices, excluding the case $m=4$, are
clear since \WM is uniquely \ECT-optimal and no protocol can
guarantee coordination in any number of rounds.

The game $\CM_4$ is the only game for which no protocol is uniquely \ECT-optimal (indeed there are uncountably many different \ECT-optimal protocols). Moreover, no protocol can guarantee coordination in this game. Based on the analysis on the other choice matching games with an even number of choices, one could possibly argue that players would naturally follow \WM also here since it is uniquely \ECT-optimal elsewhere and one of the \ECT-optimal protocols here as well. 
However, there seems to be no obvious and fully compelling
reason why \WM should be preferred to the other \ECT-optimal protocols.

The games $\CM_m$ for \emph{odd} $m\geq 7$ can also be problematic since the optimal values for \ECT and \GCT are given by different (although uniquely optimal) protocols \WM and \LA, respectively. If both players do not have the same preference about which of these values to optimize (or this is not common knowledge among them), it is not clear for them whether they should follow \WM or \LA. In the cases where $m$ is very large, say $m=1001$, \WM seems more justified in practice since it is almost impossible that coordination with \WM would take more time than with \LA. But in the cases where $m$ is quite small, especially when $m=7$, \LA may seem like a more balanced option with respect to the both aspects. Recall here that the \ECT in $\CM_7$ with \LA is 3 rounds while the \ECT with \WM is only slightly less than 3, and moreover, \LA guarantees coordination in 4 rounds while \WM does not guarantee it at all.

\medskip

\medskip

\medskip

\noindent
\textbf{Acknowledgements.}\ {We thank Valentin Goranko, Lauri Hella
and Kerkko Luosto for discussions on coordination games.
Antti Kuusisto was supported by the Academy of
Finland grants 438 874 and 209 365. Raine R\"{o}nnholm was
supported by Jenny and Antti Wihuri Foundation.}


\bibliographystyle{plain}  
\bibliography{GT-Bibliography}


\end{document}